\documentclass[
a4paper,%
pagesize,%
numbers=noendperiod,%
captions=nooneline,%
abstracton%
]{scrartcl}

\newcommand*{\mytitle}{Estimators\ for\ Archimedean\ copulas\ in\ high\
  dimensions}
\newcommand*{\myauthorone}{Marius\ Hofert}
\newcommand*{\mycontactone}{RiskLab,\ Department\ of\ Mathematics,\ ETH\ Zurich,\
  8092\ Zurich,\ Switzerland,\
  \href{mailto:marius.hofert@math.ethz.ch}{\nolinkurl{marius.hofert@math.ethz.ch}}}
\newcommand*{\mycomment}{The\ author\ (Willis\ Research\ Fellow)\ thanks\ Willis\ Re\ for\ financial\ support\ while\ this\ work\ was\ being\ completed.}
\newcommand*{\myauthortwo}{Martin\ M\"achler}
\newcommand*{\mycontacttwo}{Seminar\ f\"ur\ Statistik,\ ETH\ Zurich,\
  8092\ Zurich,\ Switzerland,\
  \href{mailto:maechler@stat.math.ethz.ch}{\nolinkurl{maechler@stat.math.ethz.ch}}}
\newcommand*{\myauthorthree}{Alexander\ J.\ McNeil}
\newcommand*{\mycontactthree}{Department\ of\ Actuarial\ Mathematics\ and\ Statistics,\
  Heriot-Watt\ University,\ Edinburgh,\ EH14 4AS,\ Scotland,\
  \href{mailto:A.J.McNeil@hw.ac.uk}{\nolinkurl{A.J.McNeil@hw.ac.uk}}}
\newcommand*{\mysubject}{Article}

\usepackage[T1]{fontenc}%
\usepackage{lmodern}%
\usepackage[american]{babel}%
\usepackage{microtype}%

\usepackage[nouppercase]{scrpage2}%
\usepackage{geometry}%
\usepackage{typearea}%

\usepackage{eso-pic}%
\usepackage{rotating}%
\usepackage{amsmath}%
\usepackage{mathtools}%
\usepackage{amssymb}%
\usepackage{amsthm}%
\usepackage{etoolbox}%
\usepackage{bm}%
\usepackage{bbm}%
\usepackage{enumitem}%
\usepackage{graphicx}%
\usepackage{grffile}%
\usepackage{tikz}%
\usepackage{wrapfig}%
\usepackage{tabularx}%
\usepackage{dcolumn}%
\usepackage{booktabs}%
\usepackage{multirow}%
\usepackage{fancyvrb}%
\usepackage{vruler}%
\usepackage{csquotes}%
\usepackage[
backend=bibtex,%
style=authoryear,%
dashed=false,%
uniquename=init,%
firstinits=true,%
hyperref=true,%
maxcitenames=2,%
maxbibnames=6,%
date=iso8601,%
urldate=iso8601%
]{biblatex}%
\usepackage[
hypertexnames=false,%
setpagesize=false,%
pdfborder={0 0 0},%
pdfstartview=Fit,%
bookmarksopen=true,%
bookmarksnumbered=true%
]{hyperref}%
\hypersetup{
  pdfauthor=\myauthorone,%
  pdftitle=\mytitle,%
  pdfsubject=\mysubject%
}

\setkomafont{pageheadfoot}{\normalfont\normalcolor\sffamily}%
\automark{section}%
\setkomafont{pagenumber}{\normalfont\normalcolor\sffamily}%
\setkomafont{captionlabel}{\normalfont\normalcolor\sffamily\bfseries}%

\setlist{%
  align=left,%
  labelsep=*,%
  leftmargin=*,%
  topsep=1mm,%
  itemsep=0mm%
}
\setlist[itemize,1]{label={\protect\rule[0.18em]{0.36em}{0.36em}\ }}%
\setlist[itemize,2]{label={\protect\raisebox{0.12em}{\resizebox{0.48em}{0.48em}{$\blacktriangleright$}}\ }}%
\setlist[itemize,3]{label={\protect\rule[0.32em]{0.62em}{0.08em}\ }}%
\setlist[enumerate,1]{label=\arabic*)}%
\setlist[enumerate,2]{label=\arabic{enumi}.\arabic*)}%
\setlist[enumerate,3]{label=\arabic{enumi}.\arabic{enumii}.\arabic*)}%

\makeatletter
\newcommand\myisodate{\number\year-\ifcase\month\or 01\or 02\or 03\or 04\or 05\or 06\or 07\or 08\or 09\or 10\or 11\or 12\fi-\ifcase\day\or 01\or 02\or 03\or 04\or 05\or 06\or 07\or 08\or 09\or 10\or 11\or 12\or 13\or 14\or 15\or 16\or 17\or 18\or 19\or 20\or 21\or 22\or 23\or 24\or 25\or 26\or 27\or 28\or 29\or 30\or 31\fi}%
\makeatother
\newcolumntype{d}[2]{D{.}{.}{#1.#2}}%
\newcommand*{\abstractnoindent}{}%
\let\abstractnoindent\abstract
\renewcommand*{\abstract}{\let\quotation\quote\let\endquotation\endquote
  \abstractnoindent}
\deffootnote[1em]{1em}{1em}{\textsuperscript{\thefootnotemark}}%

\DeclareNameAlias{sortname}{last-first}%
\DefineBibliographyExtras{american}{\DeclareQuotePunctuation{}}%
\renewbibmacro*{volume+number+eid}{%
  \setunit*{\addcomma\space}%
  \printfield{volume}%
  \printfield{number}}
\DeclareFieldFormat*{number}{(#1)}
\DeclareFieldFormat*{title}{#1}%
\DeclareFieldFormat{doi}{%
  \ifhyperref
    {\href{http://dx.doi.org/#1}{\nolinkurl{doi:#1}}}%
    {\nolinkurl{doi:#1}}}%
\renewbibmacro*{in:}{}%
\DeclareFieldFormat{isbn}{ISBN #1}%
\DeclareFieldFormat{pages}{#1}%
\DeclareFieldFormat{url}{\url{#1}}%
\DeclareFieldFormat{urldate}{\mkbibparens{#1}}%
\bibliography{./mybib}%
\renewcommand*{\cite}[2][]{\textcite[#1]{#2}}

\newif\ifstarttheorem
\newtheoremstyle{mythmstyle}%
{0.5em}%
{0.5em}%
{}%
{}%
{\sffamily\bfseries\global\starttheoremtrue}%
{}%
{\newline}%
{\thmname{#1}\ \thmnumber{#2}\ \thmnote{(#3)}}%
\theoremstyle{mythmstyle}%
\newtheorem{definition}{Definition}[section]%

\newtheorem{lemma}[definition]{Lemma}

\newtheorem{example}[definition]{Example}

\renewcommand*\proofname{Proof}
\makeatletter%
\renewenvironment{proof}[1][\proofname]{\par
  \pushQED{\qed}%
  \normalfont\topsep2\p@\@plus2\p@\relax
  \trivlist
\item[\hskip\labelsep
  \sffamily\bfseries #1]\mbox{}\hfill\\*\ignorespaces
}{%
  \popQED\endtrivlist\@endpefalse
}
\makeatother

\makeatletter
\preto\itemize{%
  \if@inlabel
    \ifstarttheorem
      \mbox{}\par\nobreak\vskip\glueexpr-\parskip-\baselineskip+0.3em\relax\hrule\@height\z@
      \global\starttheoremfalse
    \fi
  \fi}
\preto\enditemize{\global\starttheoremfalse}
\makeatother
\makeatletter
\preto\enumerate{%
  \if@inlabel
    \ifstarttheorem
      \mbox{}\par\nobreak\vskip\glueexpr-\parskip-\baselineskip+0.3em\relax\hrule\@height\z@
      \global\starttheoremfalse
    \fi
  \fi}
\preto\endenumerate{\global\starttheoremfalse}
\makeatother

\newcommand*{\R}{\textsf{R}}
\newcommand*{\pkg}[1]{\texttt{#1}}
\newcommand*{\code}[1]{\texttt{#1}}

\newcommand{\T}{^{\top}}

\newcommand*{\I}{\mathbbm{1}}
\newcommand*{\IN}{\mathbbm{N}}

\newcommand*{\IR}{\mathbbm{R}}

\newcommand*{\IP}{\mathbbm{P}}
\newcommand*{\IE}{\mathbbm{E}}

\newcommand*{\arginf}{\operatorname*{arginf}}
\newcommand*{\argsup}{\operatorname*{argsup}}

\newcommand*{\sign}{\operatorname*{sign}}

\newcommand*{\Li}[1]{\sideset{}{_{#1}}{\operatorname*{Li}}}
\renewcommand*{\th}{\bm{\theta}}

\newcommand*{\psii}{{\psi^{-1}}}
\newcommand*{\psis}[2]{{\psi_{#1}^{#2}}}
\newcommand*{\psiis}[1]{{\psi_{#1}^{-1}}}

\newcommand*{\LS}{\mathcal{LS}}
\newcommand*{\LSi}{\LS^{-1}}

\newcommand*{\U}{\operatorname*{U}}

\newcommand*{\MDE}[2]{{\text{MDE}_{#1}^{\text{#2}}}}
\newcommand*{\var}[3]{{{#1}^{\text{#2}}_{#3}}}

\begin{document}
\thispagestyle{plain}
\begin{center}
  \sffamily
  {\bfseries\LARGE\mytitle\par}
  \bigskip\smallskip
  {\Large\myauthorone\footnote{\mycontactone. \mycomment},\ \myauthortwo\footnote{\mycontacttwo},\ \myauthorthree\footnote{\mycontactthree}\par
    \bigskip
    \myisodate\par}
\end{center}
\par\bigskip
\begin{abstract}
  The performance of known and new parametric estimators for Archimedean copulas
  is investigated, with special focus on large dimensions and numerical
  difficulties. In particular, method-of-moments-like estimators based on
  pairwise Kendall's tau, a multivariate extension of Blomqvist's beta, minimum
  distance estimators, the maximum-likelihood estimator, a simulated
  maximum-likelihood estimator, and a maximum-likelihood estimator based on the
  copula diagonal are studied. Their performance is compared in a large-scale
  simulation study both under known and unknown margins (pseudo-observations),
  in small and high dimensions, under small and large dependencies, various
  different Archimedean families and sample sizes. High dimensions up to one
  hundred are considered for the first time and computational problems arising
  from such large dimensions are addressed in detail. All methods are
  implemented in the open source \R{} package \pkg{copula} and can thus be
  easily accessed and studied.
\end{abstract}
\minisec{Keywords}
Archimedean copulas, parameter estimation, Kendall's tau, Blomqvist's beta, minimum distance estimators, (diagonal/simulated) maximum-likelihood estimation.
\minisec{MSC2010}
62H12, 62F10, 62H99, 62H20, 65C60.%

\section{Introduction}
A \textit{copula} is a multivariate distribution function with
standard uniform univariate margins. An important class of copulas,
known as \textit{Archimedean copulas}, is given by
\begin{align*}
  C(\bm{u})=\psi(\psii(u_1)+\dots+\psii(u_d)),\ \bm{u}\in[0,1]^d,
\end{align*}
with \textit{generator} $\psi$. In practical applications, $\psi$ belongs to a
parametric family $(\psi_{\th})_{\th\in\Theta}$ whose parameter vector $\th$ needs to be estimated. The aims of this paper are two-fold:
\begin{enumerate}
\item To carry out a large-scale comparative study of estimation methods for Archimedean copulas for the first time, both under known and unknown margins (pseudo-observations);
\item To focus on the performamce of estimators in high dimensions, where
  considerable computational challenges (which are also addressed) have to be overcome.
\end{enumerate}

Although Archimedean copulas are exchangeable and therefore often criticized by
the scientific community because of this limitation, they are frequently used in practice; see
\cite{embrechtshofert2011c} for a discussion. Also, from a
theoretical point of view, they often serve as building blocks for more flexible
and asymmetric dependence structures (for example,
vine copulas, nested Archimedean copulas, Archimedean Sibuya copulas, Khoudraji-
or Liebscher-transformed copulas). The questions we address in this paper also
affect these (and other) dependence structures, already in much smaller dimensions
such as two to five, and have led to wrong statements
in the literature and inaccuracies as well as errors in the corresponding computations. Our accurate
computations allow us to investigate Archimedean copulas even in high
dimensions such as one hundred. To the best of our knowledge, estimating
Archimedean copulas in such large dimensions has not been considered
before (and rarely for copulas in general). As will become clear from carefully
reading this work, this is not merely another computational study. Considerable
amount of time has gone into research on how the presented estimators can be
accurately computed (including tests with high-precision arithmetic to verify
the results) and the computational power required to conduct the studies has
been high. It is more than likely that issues of this type become more important
in the future as copula models in higher dimensions become more and more of
interest, not only for practitioners. Our computations will also point out interesting (as partly
surprising) results, which might lead to further research in this direction.

There are several known approaches for estimating bivariate parametric
Archimedean copula families. Assuming the copula density to exist,
maximum-likelihood estimation is one option; see \cite{genestghoudirivest1995}
or \cite{tsukahara2005}. Another estimator resembles the method-of-moments
estimator and consists of choosing the copula parameter such that a certain
dependence measure, for example, Kendall's tau, equals its empirical
counterpart; see \cite{genestrivest1993}. Although there is no theoretical
justification for applying this method in more than two dimensions, using the
mean of pairwise empirical Kendall's taus and estimating the copula parameter
such that the population version of Kendall's tau equals this mean also appears
in the literature; see \cite{berg2009} or \cite{savutrede2010}. A similar but
different estimator is applied in \cite{kojadinovicyan2010a}. Another method in
higher dimensions based on the moments of the Kendall distribution function is
given in \cite{brahiminecir2011}. Other estimation methods include approximating
the probability integral transform with splines and using a minimum distance
approach between this distribution function and an empirical counterpart; see
\cite{dimitrovakaishevpenev2008}. Splines also appear in \cite{lambert2007} for
approximating a certain ratio involving the generator of the Archimedean copula
to be estimated. \cite{tsukahara2005} considers minimum distance estimators
based on Cram\'er-von-Mises or Kolmogorov-Smirnov distances and compares their
performance to rank approximate Z-estimators in a simulation study involving the
bivariate Archimedean Clayton, Frank, and Gumbel copula. Another estimation
procedure in the bivariate case is given by \cite{quzhoushen2010} based on
minimizing a Cram\'er-von Mises distance between the empirical distribution
function of a certain univariate random sample and the standard uniform
distribution. The approach described in \cite{stephenson2009} in the context of
extreme-value distributions can be applied for estimating the parameter of a
Gumbel copula in a Bayesian setup. A non-parametric estimation procedure is
introduced in \cite{genestneslehovaziegel2011}. For more general information
concerning copula parameter or copula density estimation in parametric and
(especially) non-parametric set-ups, see
\cite{charpentierfermanianscaillet2007}.

In this work, we compare several known and new parametric estimators
for Archimedean copulas both under known and unknown margins (the
margins being non-parametrically estimated and thus replaced by
pseudo-observations). In the large-scale simulation study carried out, we compare the following
estimators based on well-known one-parameter generators (for two-parameter
families, see \cite{hofertmaechlermcneil2012a}):
\begin{enumerate}
\item We consider the method-of-moments estimator based on averaged pairwise sample versions of Kendall's tau. We also consider the average of pairwise Kendall's tau estimators.
\item We apply a multivariate version of the measure of concordance
  known as Blomqvist's beta for estimating Archimedean
  copulas. Blomqvist's beta has the advantage of being given
  explicitly in terms of the copula. Similar to the method-of-moments
  estimation procedure introduced by \cite{genestrivest1993}, the
  copula parameters are estimated such that the population and sample version of Blomqvist's beta coincide.
\item We present several minimum distance estimators for estimating Archimedean copulas. Recently, a transformation of random variables following an Archimedean copula to uniform random variables (similar to Rosenblatt's transformation but simpler to compute) was introduced by \cite{heringhofert2012}. The minimum distance estimators presented here estimate the parameters as the minimum of certain Cram\'er-von-Mises or Kolmogorov-Smirnov distances based on the transformation of \cite{heringhofert2012}.
\item We consider maximum-likelihood estimation. Although the density of an Archimedean copula has an explicit form in theory, deriving and evaluating the required derivatives is known to be challenging from both a theoretical and a numerical perspective, especially in large dimensions. As mentioned below, computations based on computer algebra systems often fail already in low dimensions or require high precision (and are therefore too slow to be applied, for example, in large-scale simulation studies). We present explicit formulas for the densities of well-known Archimedean families and efficiently evaluate them. These results are based on the recent findings of \cite{hofertmaechlermcneil2012a}.
\item We introduce a simulated maximum-likelihood estimator to
  estimate Archimedean copulas. This estimator can be applied if the
  generator derivatives cannot be evaluated accurately but the copula is easy to sample.
\item We present maximum-likelihood estimation based on the diagonal of the Archimedean copula. The main advantage is that the resulting estimation method is comparably easy and fast to apply in virtually any dimension.
\end{enumerate}

The paper is organized as follows. In Section \ref{sec.ac}, we briefly recall the notion of Archimedean copulas. Section \ref{sec.est} introduces and presents the different estimators investigated in this work. Section \ref{sec.sim} contains the large-scale simulation carried out. Section \ref{sec.num} addresses numerical issues when working in large dimensions and provides solutions to some of the problems mentioned. Section \ref{sec.con} concludes.

\section{Archimedean copulas}\label{sec.ac}
\begin{definition}
  An \textit{(Archimedean) generator} is a continuous, decreasing function $\psi:[0,\infty]\to[0,1]$ which satisfies $\psi(0)=1$, $\psi(\infty)=\lim_{t\to\infty}\psi(t)=0$, and which is strictly decreasing on $[0,\inf\{t:\psi(t)=0\}]$. A $d$-dimensional copula $C$ is called \textit{Archimedean} if it permits the representation
  \begin{align}
    C(\bm{u}) = \psi(t(\bm{u})),\quad\text{where}\quad t(\bm{u})=\sum_{j=1}^d\psii(u_j),\quad \bm{u}\in[0,1]^d,\label{C}
  \end{align}
  for some generator $\psi$ with inverse $\psii:[0,1]\to[0,\infty]$, where $\psii(0)=\inf\{t:\psi(t)=0\}$.
\end{definition}

\cite{mcneilneslehova2009} show that a generator defines
an Archimedean copula if and only if $\psi$ is
$d$\textit{-monotone}, that is, $\psi$ is continuous on
$[0,\infty]$, admits derivatives up to the order $d-2$
satisfying $\smash[t]{(-1)^k\frac{d^k}{dt^k}\psi(t)\ge 0}$ for
all $k\in\{0,\dots,d-2\}$, $t\in(0,\infty)$, and
$\smash[t]{(-1)^{d-2}\frac{d^{d-2}}{dt^{d-2}}\psi(t)}$ is
decreasing and convex on $(0,\infty)$. We mainly assume $\psi$
to be \textit{completely monotone}, meaning that $\psi$ is
continuous on $[0,\infty]$ and
$\smash[t]{(-1)^k\frac{d^k}{dt^k}\psi(t)\ge 0}$ for all
$k\in\IN_{0}$, $t\in(0,\infty)$, so that $\psi$ is the
Laplace-Stieltjes transform $\LS[F]$ of a distribution
function $F$ on the positive real line; see Bernstein's
Theorem in \cite[p.\ 439]{feller1971}. The class of all
such generators is denoted by $\Psi_\infty$ and it is clear
that a $\psi\in\Psi_\infty$ generates an Archimedean copula in
any dimensions $d$.

There are several known parametric Archimedean generators
(see, for example, \cite[pp.\ 116]{nelsen2006}) also
referred to as \textit{Archimedean families}. Among the most
widely used in applications are those of Ali-Mikhail-Haq
(A), Clayton (C), Frank (F), Gumbel (G), and
Joe (J). We will consider these generators as working
examples; see Table~\ref{tab.gen} which also includes
population versions of Kendall's tau for these families. Here,
$D_1(\theta)=\int_0^\theta t/(\exp(t)-1)\,dt/\theta$ denotes
the \textit{Debye function of order one}. Detailed information
about the distribution functions $F$ corresponding to the
given generators can be found in \cite{hofert2011b} and
references therein.
\begin{table}[htbp]
  \centering
  \begin{tabularx}{\textwidth}{@{\extracolsep{\fill}}c@{\extracolsep{0mm}}c@{\extracolsep{0mm}}c@{\extracolsep{-7mm}}c}
    \toprule
    \multicolumn{1}{c}{Family}&\multicolumn{1}{c}{Parameter}&\multicolumn{1}{c}{$\psi(t)$}&\multicolumn{1}{c}{$\tau$}\\
    \midrule
    A&$\theta\in[0,1)$&$(1-\theta)/(\exp(t)-\theta)$&$1-2(\theta+(1-\theta)^2\log(1-\theta))/(3\theta^2)$\\
    C&$\theta\in(0,\infty)$&$(1+t)^{-{1/\theta}}$&$\theta/(\theta+2)$\\
    F&$\theta\in(0,\infty)$&$-\log\bigl(1-(1-e^{-\theta})\exp(-t)\bigr)/\theta$&$1+4(D_1(\theta)-1)/\theta$\\
    G&$\theta\in[1,\infty)$&$\exp(-t^{1/\theta})$&$(\theta-1)/\theta$\\
    J&$\theta\in[1,\infty)$&$1-(1-\exp(-t))^{1/\theta}$&$1-4\sum_{k=1}^\infty 1/(k(\theta k+2)(\theta(k-1)+2))$\\
    \bottomrule
  \end{tabularx}
  \caption{Well-known one-parameter Archimedean generators $\psi$ with
    corresponding Kendall's tau. The range of attainable Kendall's tau is
    $(0,1/3)$ for A, $(0,1)$ for C and F, and $[0,1)$ for G and J.}
  \label{tab.gen}
\end{table}

\section{Estimation methods for Archimedean copulas}\label{sec.est}
Assume that we have given realizations $\bm{x}_i$,
$i\in\{1,\dots,n\}$, of independent and identically distributed
(i.i.d.) random vectors $\bm{X}_i$, $i\in\{1,\dots,n\}$, from a joint
distribution function $H$ with known margins $F_j$,
$j\in\{1,\dots,d\}$, Archimedean copula $C$ generated by $\psi$, and
corresponding density $c$. The generator $\psi$ is assumed to belong
to a parametric family $(\psis{\th}{})_{\th\in\Theta}$ with parameter
vector $\th\in\Theta\subseteq\IR^p$, $p\in\IN$, and the true but
unknown vector is $\th_0$ (similarly, $C=C_{\th_0}$ and
$c=c_{\th_0}$). If the margins $F_j$, $j\in\{1,\dots,d\}$, are known,
$u_{ij}=F_j(x_{ij})$, $i\in\{1,\dots,n\}$, $j\in\{1,\dots,d\}$, is a
random sample from $C$. In practice, the margins are typically unknown
and must be estimated parameterically or non-parametrically. In the following, whenever working under unknown margins, we will assume the latter approach and thus consider the \textit{pseudo-observations}
\begin{align}
  \hat{u}_{ij}=\frac{n}{n+1}\hat{F}_{n,j}(x_{ij})=\frac{r_{ij}}{n+1},\label{pobs}
\end{align}
where $\hat{F}_{n,j}$ denotes the \textit{empirical distribution function} corresponding to the $j$th margin and $r_{ij}$ denotes the \textit{rank} of $x_{ij}$ among all $x_{ij}$, $i\in\{1,\dots,n\}$.

For estimating $\th_0$, we now present several methods, some of which are new. We give the formulas in terms of a random sample $\bm{U}_i$, $i\in\{1,\dots,n\}$, from $C$. In Section \ref{sec.sim}, this random sample is replaced either by realizations $\bm{u}_i$, $i\in\{1,\dots,n\}$ (when working under known margins) or by the pseudo-observations $\hat{\bm{u}}_i$, $i\in\{1,\dots,n\}$, when working under unknown margins.

\subsection{Pairwise Kendall's tau}\label{sec.tau}
\textit{Kendall's tau} is defined to be
\begin{align*}
  \tau=\IE[\sign((X_1-X_1^\prime)(X_2-X_2^\prime))],
\end{align*}
where $(X_1,X_2)\T$ is a vector of two continuously distributed random variables, $(X_1^\prime,X_2^\prime)\T$ is an independent copy of $(X_1,X_2)\T$, and $\sign(x)=\I_{(0,\infty)}(x)-\I_{(-\infty,0)}(x)$ denotes the \textit{signum function}. Kendall's tau is a measure of concordance (see \cite{scarsini1984}) and therefore measures the strength of association (as a number in $[-1,1]$) between large values of one variable and large values of the other. Note that Archimedean copulas with generator $\psi\in\Psi_\infty$ are positive lower orthant dependent, thus Kendall's tau always lies in $[0,1]$ for such copulas; see, for example, \cite[pp.\ 59]{hofert2010c}. Kendall's tau has an obvious estimator, referred to as the \textit{sample version of Kendall's tau}. Based on the random sample $\bm{U}_i=(U_{i1},U_{i2})\T$, $i\in\{1,\dots,n\}$, it is given by
\begin{align*}
  \hat{\tau}_n=\binom{n}{2}^{-1}\ \sum_{\mathclap{1\le i_1<i_2\le n}}\sign{((U_{i_11}-U_{i_21})(U_{i_12}-U_{i_22}))}.
\end{align*}%
It can also be estimated directly from the bivariate sample $\bm{X}_i$, $i\in\{1,\dots,n\}$.

If $C$ is a bivariate Archimedean copula generated by a twice continuously differentiable generator $\psi$ with $\psi(t)>0$ for all $t\in[0,\infty)$, Kendall's tau can be represented in semi-closed form as
\begin{align*}
  \tau=1+4\int_0^1\frac{\psii(t)}{(\psii(t))^\prime}\,dt=1-4\int_0^\infty t(\psi^\prime(t))^2\,dt
\end{align*}
(see \cite[p.\ 91]{joe1997}) which can often be computed explicitly; see Table~\ref{tab.gen}.

\cite{genestrivest1993} introduce a method-of-moments estimator for bivariate one-parameter Archimedean copulas based on Kendall's tau. The copula parameter $\theta_0\in\Theta\subseteq\IR$ is estimated by $\hat{\theta}_n$ such that
\begin{align*}
  \tau(\hat{\theta}_n)=\hat{\tau}_n,
\end{align*}
where $\tau(\theta)$ denotes Kendall's tau of the corresponding Archimedean family viewed as a function of the parameter $\theta\in\Theta\subseteq\IR$. In other words,
\begin{align}
  \hat{\theta}_n=\tau^{-1}(\hat{\tau}_n),\label{est.tau}
\end{align}
assuming the inverse $\tau^{-1}$ of $\tau$ exists. This estimation method
obviously only applies to one-parameter families. Otherwise, the set of all
parameters with equal Kendall's tau is a level curve and so Kendall's tau cannot
be uniquely inverted. If (\ref{est.tau}) has no solution, this
estimation method does not lead to an estimator. Note that unless there is an explicit
form for $\tau^{-1}$, $\hat{\theta}_n$ is computed by numerical root finding.

\cite{berg2009} and \cite{savutrede2010} apply this method to data of dimension $d>2$ by using pairwise sample versions of Kendall's tau. If $\hat{\tau}_{n,j_1j_2}$ denotes the sample version of Kendall's tau between the $j_1$th and $j_2$th data column, then $\theta$ is estimated by
\begin{align}
  \hat{\theta}_n=\tau^{-1}\Biggl(\binom{d}{2}^{-1}\ \sum_{\mathclap{1\le j_1<j_2\le d}}\hat{\tau}_{n,j_1j_2}\Biggr).\label{tau.tau}
\end{align}
We denote this estimator or estimation method by
$\tau_{\bar{\hat{\tau}}}$. Intuitively, the parameter is chosen such
that Kendall's tau equals the average over all pairwise sample
versions of Kendall's tau. Note that properties of this estimator are
not known and also not easy to derive since the average is taken over
dependent data columns. In particular, although
$\binom{d}{2}^{-1}\sum_{1\le j_1<j_2\le d}\hat{\tau}_{n,j_1j_2}$ is
unbiased for $\tau(\theta_0)$, the estimator in (\ref{tau.tau}) need not be unbiased for $\theta_0$.

Another ``pairwise'' estimator can be obtained by first computing the $\binom{d}{2}$ pairwise estimators as given in (\ref{est.tau}) and then average over the estimators, that is,
\begin{align*}
  \hat{\theta}_n=\binom{d}{2}^{-1}\ \sum_{\mathclap{1\le j_1<j_2\le d}}\tau^{-1}(\hat{\tau}_{n,j_1j_2}).
\end{align*}
This unbiased estimator can be found in \cite{kojadinovicyan2010a}; see, for example, the function \code{fitCopula(, method=``itau'')} in the \R{} package \pkg{copula}. We denote it or the corresponding estimation method by $\tau_{\bar{\hat{\theta}}}$.

\subsection{Blomqvist's beta}\label{sec.blom}
\textit{Blomqvist's beta} (see, for example, \cite[p.\ 182]{nelsen2006}) is also a measure of concordance. In the bivariate case with $X_j\sim F_j$, $j\in\{1,2\}$, it is defined by
\begin{align*}
  \beta=\IP((X_1-F_1^-(1/2))(X_2-F_2^-(1/2))>0)-\IP((X_1-F_1^-(1/2))(X_2-F_2^-(1/2))<0)
\end{align*}
and therefore measures the probability of falling into the first or third quadrant minus the probability of falling into the second or fourth quadrant, the quadrants being defined by the medians $F_j^-(1/2)$, $j\in\{1,2\}$. This measure can be expressed in terms of the copula of $(X_1,X_2)\T$. It also allows for a natural generalization to $d>2$, given by
\begin{align*}
  \beta=\frac{2^{d-1}}{2^{d-1}-1}(C(1/2,\dots,1/2)+\hat{C}(1/2,\dots,1/2)-2^{1-d});
\end{align*}
see, for example, \cite{schmidschmidt2007}. Here, $\hat{C}$ denotes the survival copula corresponding to $C$. For Archimedean copulas as given in (\ref{C}), Blomqvist's beta is easily seen to be
\begin{align}
  \beta=\frac{2^{d-1}}{2^{d-1}-1}\biggl(\psi(d\psii(1/2))+\biggl(\,\sum_{j=0}^d\binom{d}{j}(-1)^j\psi(j\psii(1/2))\biggr)-2^{1-d}\biggr).\label{beta}
\end{align}
Given the random sample $\bm{U}_i$, $i\in\{1,\dots,n\}$, the \textit{sample version of Blomqvist's beta} is given by
\begin{align}
  \hat{\beta}_n=\frac{2^{d-1}}{2^{d-1}-1}\biggl(\frac{1}{n}\sum_{i=1}^n\biggl(\prod_{j=1}^d\I_{\{U_{ij}\le1/2\}}+\prod_{j=1}^d\I_{\{U_{ij}>1/2\}}\biggr)-2^{1-d}\biggr)\label{beta.hat}
\end{align}
For asymptotic properties of $\hat{\beta}_n$, see \cite{schmidschmidt2007}.

A method-of-moments estimator based on Blomqvist's beta can be obtained via
\begin{align*}
  \hat{\theta}_n=\beta^{-1}(\hat{\beta}_n),
\end{align*}
where $\beta(\theta)$ denotes $\beta$ as a function of the parameter
$\theta\in\Theta\subseteq\IR$. We denote this estimator or estimation method by
$\beta$. As for Kendall's tau, this estimation method only applies to the
one-parameter case. Typically, $\hat{\theta}_n$ is computed via numerical root finding.

\subsection{Minimum distance estimation}\label{sec.mde}
\cite{heringhofert2012} present a transformation for Archimedean copulas that is analogous to Rosenblatt's transform but simpler to compute. Consider a $d$-monotone generator $\psi$ and let $\bm{U}$ follow the Archimedean copula $C$ with generator $\psi$. Furthermore, let the \textit{Kendall distribution function} $K$ (that is, the distribution function of the \textit{probability integral transformation} $C(\bm{U})$) be continuous. Then, the transformed random vector $\bm{U^\prime}=T_\psi(\bm{U})$ with
\begin{align}
  U_{j}^\prime=\left(\frac{\sum_{k=1}^{j}\psii(U_{k})}{\sum_{k=1}^{j+1}\psii(U_{k})}\right)^{j},\ j\in\{1,\dots,d-1\},\ U_{d}^\prime=K(C(\bm{U}))\label{trafo}
\end{align}
follows a uniform distribution on $[0,1]^d$, denoted by $\bm{U^\prime}\sim\U[0,1]^d$. Note that if $\psi\in\Psi_\infty$, then $K(t)=\sum_{k=0}^{d-1}\frac{\psi^{(k)}(\psii(t))}{k!}(-\psii(t))^k$; see \cite{barbegenestghoudiremillard1996} or \cite{mcneilneslehova2009}. The transformation (\ref{trafo}) allows one to easily derive a minimum distance estimator. First, one transforms the random vectors $\bm{U}_i$, $i\in\{1,\dots,n\}$, with $T_\psi$ and then minimizes a ``distance'' between the transformed variates and the multivariate uniform distribution. This could be achieved, for example, with the statistics $S_n^{(B)}$ or $S_n^{(C)}$ used by \cite{genestremillardbeaudoin2009}. For simplicity and run-time performance, however, we map the transformed variates to univariate quantities via
\begin{align*}
  Y^{\text{n}}_i=\sum_{j=1}^d(\Phi^{-1}(U^\prime_{ij}))^2\quad\text{or}\quad Y^{\text{l}}_i=\sum_{j=1}^d-\log U^\prime_{ij},\ i\in\{1,\dots,n\},
\end{align*}
where $\Phi^{-1}$ denotes the quantile function of the standard normal distribution. Such mappings to a univariate setting are known from goodness-of-fit testing; see \cite[p.\ 97]{dagostinostephens1986}. If the transformation $T_\psi(\bm{U})$ is applied with the correct parameter, then $Y^{\text{n}}_i\sim F_{\chi_d^2}$ and $Y^{\text{l}}_i\sim F_{\Gamma_d}$, $i\in\{1,\dots,n\}$, that is, $Y^{\text{n}}_i$ and $Y^{\text{l}}_i$ should follow a chi-square distribution with $d$ degrees of freedom and a $\Gamma(d,1)$ distribution, respectively. Hence, minimum distance estimators can be obtained via the Cram\'er-von Mises and Kolmogorov-Smirnov type of distances
\begin{align*}
  \hat{\th}_n^{\text{n,CvM}}&=\arginf_{\th\in\Theta}n\int_{-\infty}^\infty\bigl\lvert\hat{F}_{n,Y^{\text{n}}}(x)-F_{\chi_d^2}(x)\bigr\rvert^2\,dF_{\chi_d^2}(x)\\
&=\arginf_{\th\in\Theta}\frac{1}{12n}+\sum_{i=1}^n\biggl(\frac{2i-1}{2n}-F_{\chi_d^2}(Y_{(i)}^{\text{n}})\biggr)^2,\\
  \hat{\th}_n^{\text{n,KS}}&=\arginf_{\th\in\Theta}\sup_{x}\bigl\lvert\hat{F}_{n,Y^{\text{n}}}(x)-F_{\chi_d^2}(x)\bigr\rvert\\
&=\arginf_{\th\in\Theta}\max_{i\in\{1,\dots,n\}}\biggl\{F_{\chi_d^2}(Y_{(i)}^{\text{n}})-\frac{i-1}{n},\frac{i}{n}-F_{\chi_d^2}(Y_{(i)}^{\text{n}})\biggr\},\\
  \hat{\th}_n^{\text{l,CvM}}&=\arginf_{\th\in\Theta}n\int_{-\infty}^\infty\bigl\lvert\hat{F}_{n,Y^{\text{l}}}(x)-F_{\Gamma_d}(x)\bigr\rvert^2\,dF_{\Gamma_d}(x)\\
&=\arginf_{\th\in\Theta}\frac{1}{12n}+\sum_{i=1}^n\biggl(\frac{2i-1}{2n}-F_{\Gamma_d}(Y_{(i)}^{\text{l}})\biggr)^2,\\
  \hat{\th}_n^{\text{l,KS}}&=\arginf_{\th\in\Theta}\sup_{x}\bigl\lvert\hat{F}_{n,Y^{\text{l}}}(x)-F_{\Gamma_d}(x)\bigr\rvert\\
&=\arginf_{\th\in\Theta}\max_{i\in\{1,\dots,n\}}\biggl\{F_{\Gamma_d}(Y_{(i)}^{\text{l}})-\frac{i-1}{n},\frac{i}{n}-F_{\Gamma_d}(Y_{(i)}^{\text{l}})\biggr\},
\end{align*}
where $\hat{F}_{n,Y^{\text{n}}}$ and $\hat{F}_{n,Y^{\text{l}}}$ denote the empirical distribution functions based on $(Y^{\text{n}}_i)_{i\in\{1,\dots,n\}}$ and $(Y^{\text{l}}_i)_{i\in\{1,\dots,n\}}$, respectively, and $Y_{(i)}^{\text{n}}$ and $Y_{(i)}^{\text{l}}$, $i\in\{1,\dots,n\}$, denote the order statistics of $(Y^{\text{n}}_i)_{i\in\{1,\dots,n\}}$ and $(Y^{\text{l}}_i)_{i\in\{1,\dots,n\}}$, respectively. We denote these four estimators or estimation methods by $\MDE{\chi}{CvM}$, $\MDE{\chi}{KS}$, $\MDE{\Gamma}{CvM}$, and $\MDE{\Gamma}{KS}$, respectively.

In large dimensions, one can omit the possibly costly computation of $U_{d}^\prime$ and work with $U_{j}^\prime$, $j\in\{1,\dots,d-1\}$, only, see \cite[pp.\ 52]{hering2011}. Note that minimum distance estimators naturally also work for $p\ge2$, that is, parameter vectors $\th\in\Theta\subseteq\IR^p$.

\subsection{Maximum-likelihood estimation}\label{sec.mle}
According to \cite{mcneilneslehova2009}, an Archimedean copula $C$ admits a density $c$ if and only if $\psi^{(d-1)}$ exists and is absolutely continuous on $(0,\infty)$. In this case, $c$ is given by
\begin{align}
  c(\bm{u})=\psi^{(d)}(t(\bm{u}))\prod_{j=1}^d(\psii)^\prime(u_j) =
  \frac{\psi^{(d)}(t(\bm{u}))}{\prod_{j=1}^d\psi^\prime(\psii(u_j))},\ \bm{u}\in(0,1)^d,\label{c}
\end{align}
where, as in (\ref{C}), $t(\bm{u})=\sum_{j=1}^d\psii(u_j)$.
Note that for computing the log-density, it is convenient to write $c$ as
\begin{align*}
  c(\bm{u})=(-1)^d\psi^{(d)}(t(\bm{u}))\prod_{j=1}^d-(\psii)^\prime(u_j)=\frac{(-1)^d\psi^{(d)}(t(\bm{u}))}{\prod_{j=1}^d-\psi^\prime(\psii(u_j))}.
\end{align*}
Given the sample $\bm{U}_i$, $i\in\{1,\dots,n\}$, finding the maximum-likelihood estimator (MLE) usually involves solving the optimization problem
\begin{align*}
  \hat{\th}_n=\argsup_{\th\in\Theta}\sum_{i=1}^n\log c_{\th}(\bm{U}_i),
\end{align*}
where here and in the following the subscript $\th$ is used to stress the
dependence on $\th$. This requires an efficient strategy for evaluating the
(log-)density. The most important part is to know how to derive and compute
the generator derivatives. Tools like automatic differentiation, see \cite{griewankwalther2003}, might provide a solution. Recently, \cite{hofertmaechlermcneil2012a}
presented explicit formulas for all families listed in Table~\ref{tab.gen}.
The corresponding copula densities are reported here
for the reader's convenience (note that $\alpha=1/\theta$):
\begin{enumerate}
\item For the family of Ali-Mikhail-Haq,
  \begin{align*}
    c_\theta(\bm{u})=\frac{(1-\theta)^{d+1}}{\theta^2}\frac{\var{h}{A}{\theta}(\bm{u})}{\prod_{j=1}^du_j^2}
    \Li{-d}(\var{h}{A}{\theta}(\bm{u})),
  \end{align*}
  where $\Li{-s}(z)=\sum_{k=1}^\infty\frac{z^k}{k^s}$ denotes the \textit{polylogarithm of order $s$ at $z$} and $\var{h}{A}{\theta}(\bm{u})=\theta\prod_{j=1}^d\frac{u_j}{1-\theta(1-u_j)}$.
\item For the family of Clayton,
  \begin{align*}
    c_\theta(\bm{u})=\prod_{k=0}^{d-1}(\theta k+1)\biggl(\,\prod_{j=1}^du_j\biggr)^{-(1+\theta)}(1+t_\theta(\bm{u}))^{-(d+\alpha)}.
  \end{align*}
\item For the family of Frank,
  \begin{align*}
    c_\theta(\bm{u})=\biggl(\frac{\theta}{1-e^{-\theta}}\biggr)^{d-1}
    \Li{-(d-1)}(\var{h}{F}{\theta}(\bm{u}))\frac{\exp(-\theta\sum_{j=1}^du_j)}{\var{h}{F}{\theta}(\bm{u})},
  \end{align*}
  where $\var{h}{F}{\theta}(\bm{u})=(1-e^{-\theta})^{1-d}\prod_{j=1}^d(1-\exp(-\theta u_j))$.
\item\label{sec.mle.G} For the family of Gumbel,
  \begin{align*}
    c_\theta(\bm{u})=\theta^d\exp(-t_\theta(\bm{u})^{\alpha})
    \frac{\prod_{j=1}^d(-\log u_j)^{\theta-1}}{t_\theta(\bm{u})^d\prod_{j=1}^du_j}
    \var{P}{G}{d,\alpha}(t_\theta(\bm{u})^{\alpha}),
  \end{align*}
  where
  \begin{align*}
    \var{P}{G}{d,\alpha}(x)&=\sum_{k=1}^d\var{a}{G}{dk}(\alpha)x^k,\\
    \var{a}{G}{dk}(\alpha)&=(-1)^{d-k}\sum_{j=k}^d\alpha^{j}s(d,j)S(j,k)=\frac{d!}{k!}\sum_{j=1}^k\binom{k}{j}\binom{\alpha j}{d}(-1)^{d-j},\ k\in\{1,\dots,d\},
  \end{align*}
  and $s$ and $S$ denote the \textit{Stirling numbers of the first kind} and the
  \textit{second kind}, respectively, given by the recurrence relations
    \begin{align*}
      s(n+1,k)&=s(n,k-1)-ns(n,k),\\
      S(n+1,k)&=S(n,k-1)+kS(n,k),
    \end{align*}
    for all $k\in\IN$, $n\in\IN_0$, with $s(0,0)=S(0,0)=1$ and $s(n,0)=s(0,n)=S(n,0)=S(0,n)=0$ for all $n\in\IN$.
\item\label{sec.mle.J} For the family of Joe,
  \begin{align*}
    c_\theta(\bm{u})=\theta^{d-1}\frac{\prod_{j=1}^d(1-u_j)^{\theta-1}}{(1-\var{h}{J}{\theta}(\bm{u}))^{1-1/\theta}}\var{P}{J}{d,\alpha}\biggl(\frac{\var{h}{J}{\theta}(\bm{u})}{1-\var{h}{J}{\theta}(\bm{u})}\biggr),
  \end{align*}
  where
  \begin{align*}
    \var{P}{J}{d,\alpha}(x)&=\sum_{k=0}^{d-1}\var{a}{J}{dk}(\alpha)x^k,\\
    \var{a}{J}{dk}(\alpha)&=S(d,k+1)(k-\alpha)_{k},\ k\in\{0,\dots,d-1\},
  \end{align*}
  $\var{h}{J}{\theta}(\bm{u})=\prod_{j=1}^d(1-(1-u_j)^\theta)$, and $(k-\alpha)_{k}=\frac{\Gamma(k+1-\alpha)}{\Gamma(1-\alpha)}$ denotes the falling factorial.
\end{enumerate}
\begin{example}
  The left-hand side of Figure \ref{fig.logL} shows the log-likelihood of a Clayton copula based on a 100-dimensional sample of size $n=100$ with parameter $\theta_0=2$ such that the corresponding bivariate population version of Kendall's tau equals $\tau(\theta_0)=0.5$. The MLE is denoted by $\hat{\theta}_n$. The right-hand side of Figure \ref{fig.logL} shows the log-likelihood plot of a 100-dimensional Gumbel family with parameter $\theta_0=2$ such that Kendall's tau equals $\tau(\theta_0)=0.5$. Both Figures are plotted on the interval $[\tau^{-1}(\tau(\theta_0)-h),\tau^{-1}(\tau(\theta_0)+h)]$, where $h=0.05$ denotes a ``distance'' in terms of concordance. Note that evaluating the log-density of a Gumbel copula is numerically highly complicated; see Section \ref{sec.num.GJ} for more details.
  \begin{figure}[htbp]
    \centering
    \includegraphics[width=0.48\textwidth]{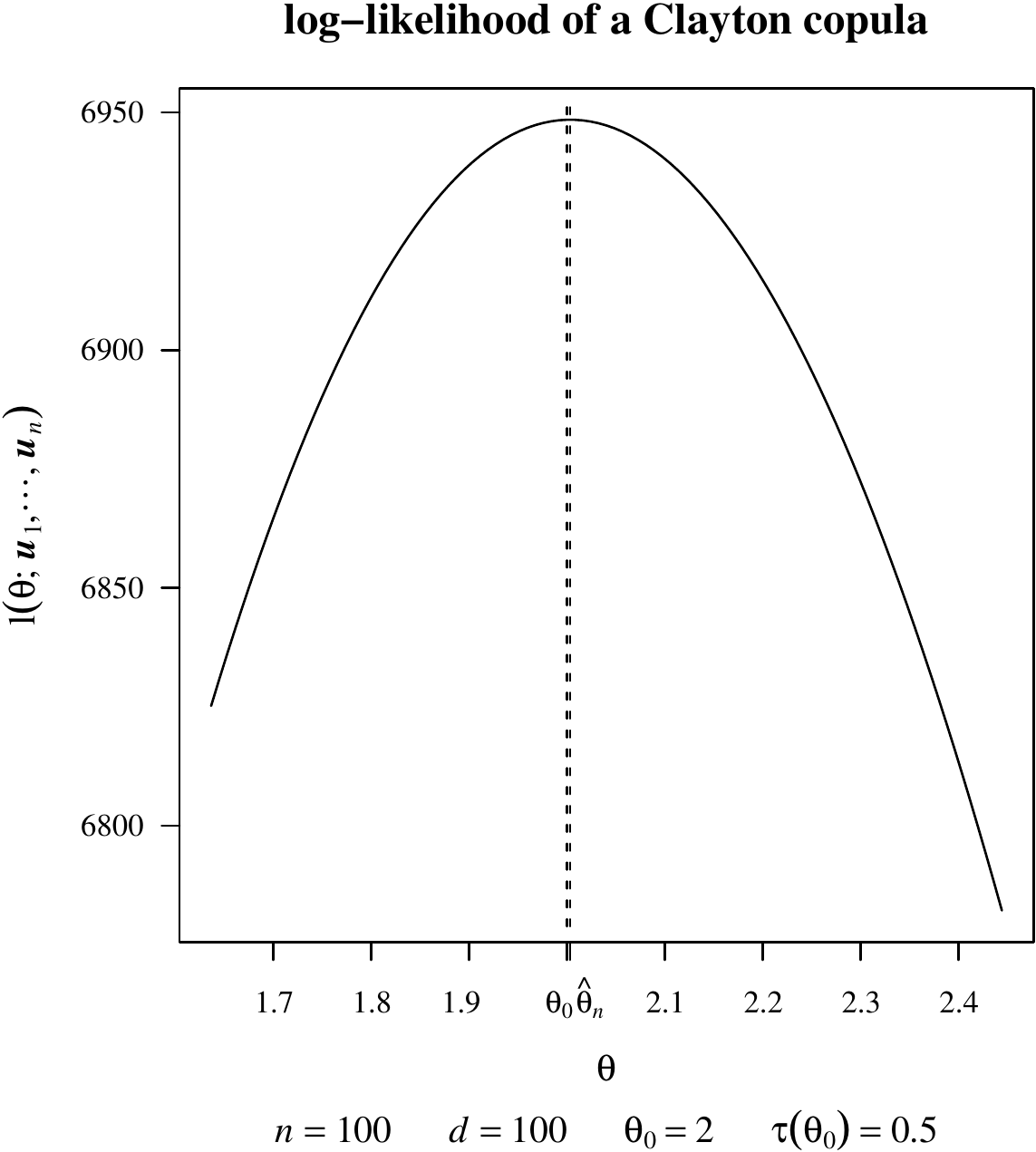}%
    \hfill
    \includegraphics[width=0.48\textwidth]{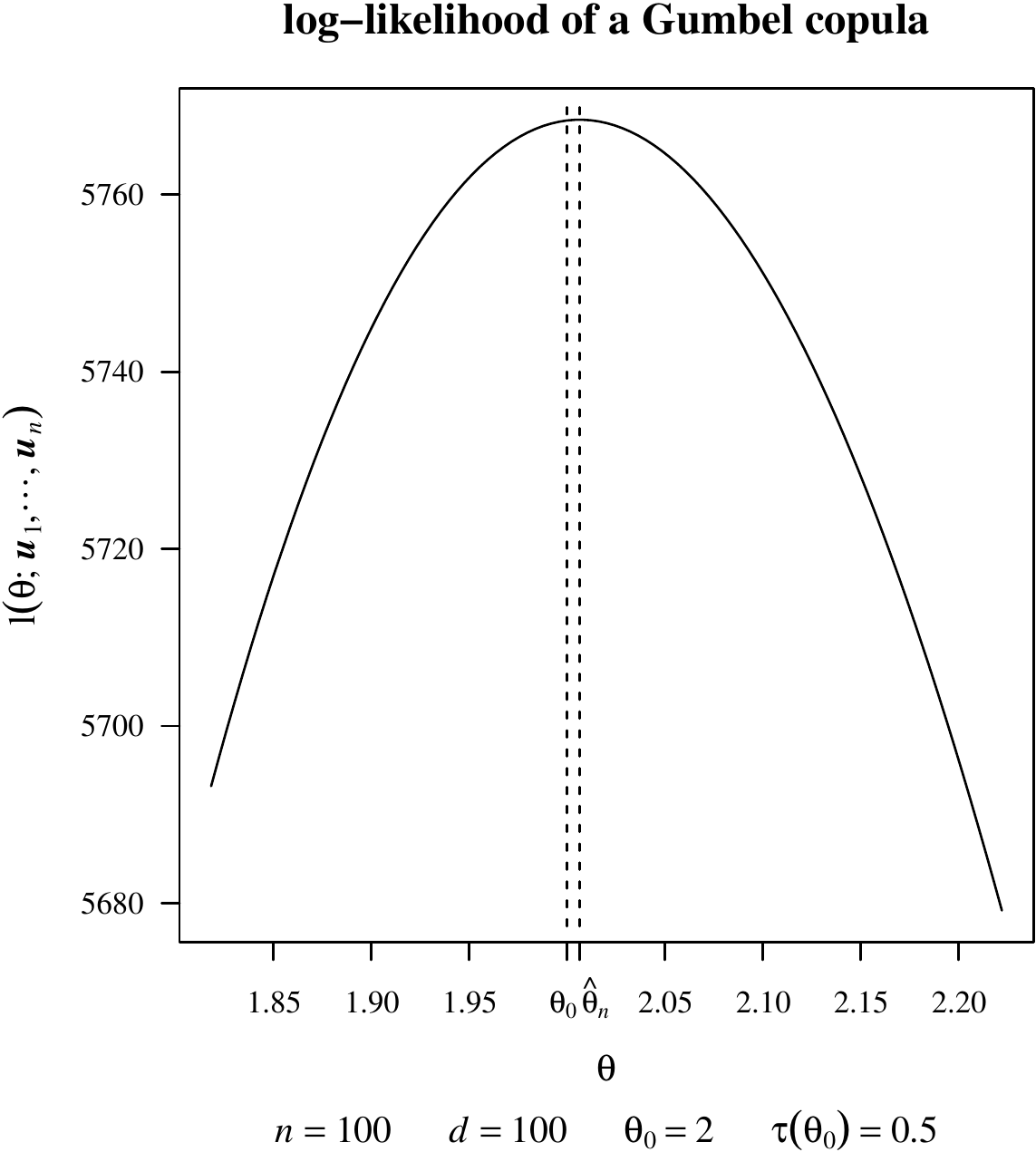}%
    \caption{Plot of the log-likelihood of a Clayton (left) and Gumbel (right) copula based on a sample of size $n=100$ in dimension $d=100$ with parameter $\theta_0=2$ such that Kendall's tau equals $0.5$.}
    \label{fig.logL}
  \end{figure}
\end{example}

\subsection{Simulated maximum-likelihood estimation}\label{sec.smle}
If the derivatives of a given Archimedean generator $\psi$ are not known explicitly one may use the fact that $\psi$ is an expectation in order to approximate the density of the generated copula. This way one can replace derivatives of higher order by just one integral (which can either be evaluated numerically or via Monte Carlo simulation). If $\psi\in\Psi_\infty$, then $\psi(t)=\LS[F](t)=\int_0^\infty\exp(-xt)\,dF(x)$, so that by differentiating under the integral sign one obtains
\begin{align*}
  (-1)^d\psi^{(d)}(t)=\int_0^\infty x^d\exp(-xt)\,dF(x)=\IE[V^d\exp(-Vt)],\ t\in(0,\infty),
\end{align*}
where $V$ has distribution function $F$. An approximation to $(-1)^d\psi^{(d)}(t)$ is thus given by
\begin{align}
  (-1)^d\psi^{(d)}(t)\approx\frac{1}{m}\sum_{k=1}^mV_k^d\exp(-V_kt),\ t\in(0,\infty),\label{der.smle}
\end{align}
where $V_k\sim F$, $k\in\{1,\dots,m\}$, are realizations of i.i.d.\ random variables following $F=\LSi[\psi]$. Instructions for how to sample $F$ for the one-parameter families in Table~\ref{tab.gen} can be found, for example, in \cite{hofert2011b}; see also \cite{hofertmaechler2011}. This method can be used to evaluate the copula density. We refer to the corresponding MLE as \textit{simulated maximum-likelihood estimator} (SMLE). Finally, note that both the MLE and the SMLE naturally apply to the multi-parameter case.

\subsection{Diagonal maximum-likelihood estimation}\label{sec.dmle}
It is well-known that the diagonal $\delta(u)=C(u,\dots,u)$ of a copula $C$ does
not uniquely determine $C$. However, it is also known that a bivariate
associative copula $C$ whose diagonal $\delta$ satisfies $\delta(u)<u$ for all
$u\in(0,1)$ is an Archimedean copula; see \cite[p.\ 113]{nelsen2006}. As far as
we are aware, this \textit{diagonal property} has not been exploited for
estimating Archimedean (or other) copulas. It suggests a simple and
straightforward estimation procedure based on the information from the copula
diagonal, described as follows. Note that the diagonal $\delta_{\th}$ of a parametric copula family $(C_{\th})_{\th\in\Theta}$ is a distribution function and that for $\bm{U}\sim C$,
\begin{align*}
  Y=\max_{1\le j\le d}U_j\sim\delta_{\th}.
\end{align*}
Based on the sample $\bm{U}_i$, $i\in\{1,\dots,n\}$, with corresponding maxima $Y_i$, $i\in\{1,\dots,n\}$, one can apply maximum-likelihood estimation to find an estimator $\hat{\th}_n$ of the parameter vector $\th_0$ via
\begin{align}
  \hat{\th}_n=\argsup_{\th\in\Theta}\sum_{i=1}^n\log\delta_{\th}^\prime(Y_i),\label{min.dmle}
\end{align}
where $\delta^\prime_{\th}$ denotes the density of the distribution function $\delta_{\th}$. We refer to the estimator $\hat{\th}_n$ as \textit{diagonal maximum-likelihood estimator} (DMLE).
For Archimedean copulas, $\delta_{\th}(u)=\psi_{\th}(d\psiis{\th}(u))$ and hence,
\begin{align}\label{diagdensity}
  \delta^\prime_{\th}(u) =
  d\psi_{\th}^\prime(d\psiis{\th}(u))(\psiis{\th})^\prime(u),\quad u\in[0,1].
\end{align}
Therefore, one advantage of the DMLE is that the degree of numerical difficulty of the optimization in (\ref{min.dmle}) (theoretically) remains rather unaffected by the dimension. For the one-parameter Gumbel family, (\ref{min.dmle}) can even be solved explicitly, the estimator $\hat{\theta}_n^{\text{G}}$ of $\theta$ being
\begin{align*}
  \hat{\theta}_n^{\text{G}}=\frac{\log d}{\log n-\log\bigl(\sum_{i=1}^n-\log Y_i\bigr)}.
\end{align*}
An adjusted estimator of the form
\begin{align*}
  \hat{\theta}_n^{\text{G},\ast}=\max\{\hat{\theta}_n^{\text{G}},1\}
\end{align*}
is then guaranteed to provide an admissible parameter estimator for Gumbel's family.

\section{A large-scale simulation study}\label{sec.sim}
In this section, we present a large-scale simulation study in which we compare the performance of the different estimators presented in Section \ref{sec.est} both under known and unknown margins (pseudo-observations). To the best of our knowledge, this is the first study of this kind also addressing large dimensions (up to 100). To be able to also include the estimators based on measures of concordance, we restrict ourselves to the one-parameter families as given in Table~\ref{tab.gen}.

\subsection{A word concerning the implementation}
The results presented in this section are based on the following computational set-up. The procedures are computationally challenging in many ways and much effort has gone into accurate and efficient implementation in \R; see Section \ref{sec.num}. The latest version of the package can be accessed via \url{http://nacopula.r-forge.r-project.org/}. The computations are carried out on the computer cluster Brutus of ETH Zurich which runs CentOS 5.4. The batch jobs are run on nodes with four quad-core AMD Opteron 8380 CPUs and 32\,GB of RAM. Apart from the physical structure of the grid, the compiler, and the programming language, note that run time also depends on other factors such as the quality of the implementation or the current load of the machine. The presented run times should therefore be viewed with this in mind.

\subsection{The experimental design}\label{sec.exp.des}
In the simulation study carried out, we consider both known and unknown
margins. For each of these cases, we generate $N=1000$ samples of size $n$ from
i.i.d.\ random vectors following a $d$-dimensional Archimedean copula with
prespecified parameter $\theta$ such that the corresponding Kendall's tau is
$\tau\in\{0.25,0.75\}$. For the case of unknown margins, we build the
pseudo-observations as given in (\ref{pobs}). Since we are mainly interested in
the behavior for different dimensions $d$, we consider $d\in\{5,20,100\}$ and
fix $n=100$ (so the data matrices considered have up to 10\,000 entries). We investigate the one-parameter families of Ali-Mikhail-Haq (only for $\tau=0.25$ since the range of admissible Kendall's tau for this family is bounded from above by $1/3$), Clayton, Frank, Gumbel, and Joe. The average pairwise Kendall's tau estimator ($\tau_{\bar{\hat{\tau}}}$), the average of Kendall's tau estimators ($\tau_{\bar{\hat{\theta}}}$), the estimation method based on Blomqvist's beta ($\beta$), the four presented minimum distance estimators ($\MDE{\chi}{CvM}$, $\MDE{\chi}{KS}$, $\MDE{\Gamma}{CvM}$, and $\MDE{\Gamma}{KS}$), maximum-likelihood (MLE), simulated maximum-likelihood (SMLE), and diagonal maximum-likelihood (DMLE) estimators are applied to estimate the parameter for each of the $N$ data sets. Finally, bias and root mean squared error (RMSE), as well as mean user run time over all $N$ replications are computed.

For the required optimizations for the estimators based on Blomqvist's beta, all minimum distance estimators, MLE, SMLE, and DMLE, we use initial intervals determined from a large range of (admissible) Kendall's tau; see the implementation of the function \code{initOpt} in the \R{} package \pkg{copula} for more details. For the minimum distance estimators to be competitive according to run time, we only include the Kendall distribution function in the transformation given in (\ref{trafo}) in the five-dimensional case, but not for higher dimensions. For applying the SMLE, we draw $10\,000$ random variates from $V\sim F=\LS[\psi]$ for each evaluation of the density of the Archimedean copula.

\subsection{Results under known margins}\label{sec.res}
Tables~\ref{tab.bias}, \ref{tab.rmse}, and \ref{tab.mut} in the appendix contain the bias (multiplied by 1000), the RMSE (multiplied by 1000), and the mean user times in milliseconds (MUT), respectively, for all investigated estimators under known margins. For each entry, the number in parentheses denotes the entry divided by the corresponding entry of the MLE column, so that the performance with respect to the MLE can easily be determined; the MLE itself thus has always 1.0 in parentheses. For the RMSEs, note that the reciprocals of the square of these numbers are also known as the (estimated) \textit{relative efficiency} of the MLE with respect to the corresponding estimator.

Figures \ref{fig.err.tau.0.25} and \ref{fig.err.tau.0.75} graphically display the square root of the absolute error via box plots obtained from the $N=1000$ replications. Due to readability, we exclude methods which perform so poorly that their boxplots dominate the scale of values. In particular, this is the case for the estimators based on Blomqvist's beta for $d=100$ for the families of Clayton, Frank, and Joe and the SMLE for Clayton and $\tau=0.25$.

The results from this study under known margins can be summarized as follows:
\begin{itemize}
\item The performance of the average of bivariate Kendall's tau estimator $\tau_{\bar{\hat{\theta}}}$ as given in the end of Section \ref{sec.tau} is very similar to the one of the averaged pairwise Kendall tau estimator $\tau_{\bar{\hat{\tau}}}$. One problem that especially $\tau_{\bar{\hat{\theta}}}$ faces is that sample versions of Kendall's tau are sometimes not in the range of tau as a function of theta. These values were then mapped to the range of admissible Kendall's tau, see the function \code{tau.checker} in \pkg{copula}. Furthermore, run time for method $\tau_{\bar{\hat{\theta}}}$ is typically larger (especially in large dimensions) than that for $\tau_{\bar{\hat{\tau}}}$, which is clear since more inversions of Kendall's tau have to be performed. Overall, $\tau_{\bar{\hat{\tau}}}$ is thus preferred. Furthermore, due to only considering pairs at a time, this estimator is quite robust against numerical difficulties. A disadvantage, however, is its large run time due to the quadratic complexity in the dimension $d$.
\item Although Blomqvist's beta can be flawlessly applied to estimate
the copula parameter $\theta$ for small and moderate dimensions $d$, this estimator shows
serious numerical problems for $d$ uniformly over
all investigated Archimedean families. One of the problems turns
out to be that both products appearing in the sample version
(\ref{beta.hat}) of Blomqvist's beta are sometimes zero, so that
$\hat{\beta}_n<0$ although $\beta\ge0$. Another problem is that
the evaluation of the survival copula at $(1/2,\dots,1/2)\T$ turns out to be numerically challenging for several families; see Section \ref{sec.num.dfrank} for more details.
\item The performance of the minimum distance estimators depends on the mapping to the one-dimensional setting applied. In particular, the estimators $\MDE{\Gamma}{CvM}$ and $\MDE{\Gamma}{KS}$ based on the logarithmic transformation to a Gamma distribution do not perform well in comparison to $\MDE{\chi}{CvM}$ and $\MDE{\chi}{KS}$. Furthermore, the minimum distance estimators based on the Kolmogorov-Smirnov distances are outperformed by those based on the Cram\'er-von Mises distances (also according to run time in most of the cases investigated, see Table \ref{tab.mut}). Note that run time for $d=5$ is larger than for $d\in\{20,100\}$ because we applied the full transformation $T_\psi$ including the Kendall distribution function $K$ in the five-dimensional case. Moreover, note that the distances (objective functions) had to be reparameterized in order for the minimum distance estimators to be computed; see Section \ref{sec.num.mde} for more details.
\item With the explicit formulas for the densities we presented, the MLE
  clearly shows the best performance under known margins. Note that the run
  times are much smaller than one would expect in comparison to other
  estimators, although, our implementation was written with focus on
  readability rather than run-time performance and thus several quantities
  are computed each time the density is evaluated. In contrast to
  statements found in the literature (see, for example, \cite{bergaas2009} or \cite{weiss2010}) this leaves no doubt that maximum-likelihood estimation is feasible in
  large dimensions and performs well; for the latter, see also
  \cite{hofertmaechlermcneil2012a} who empirically show that the mean
  squared error (MSE) satisfies
  \begin{align*}
    \text{MSE}\propto\frac{1}{nd}.
  \end{align*}
\item The SMLE also shows an incredible performance, the only exception being Clayton's family; see Section \ref{sec.num.smle} for more details. The RMSEs are close to the ones obtained by maximum-likelihood estimation. Moreover, this method is straightforward to implement given random number generators for the distribution corresponding to the generator under consideration. The drawback of this method is certainly a larger run time. Note that this could be partly reduced, for example, by using an adaptive technique in which smaller amounts of random variates are drawn during the first couple of steps the optimizer performs. Also note that with the constant use of 5000 random variates (instead of 10\,000) per density evaluation, the overall performance of the SMLE is still slightly better than those of the minimum distance estimator $\MDE{\chi}{CvM}$.
\item The advantage of the DMLE lies in its speed. Due to this fact, this estimator could be used for finding initial values for more sophisticated estimation methods.
\end{itemize}
\begin{figure}[htbp]
  \centering
  \includegraphics[height=0.97\textheight]{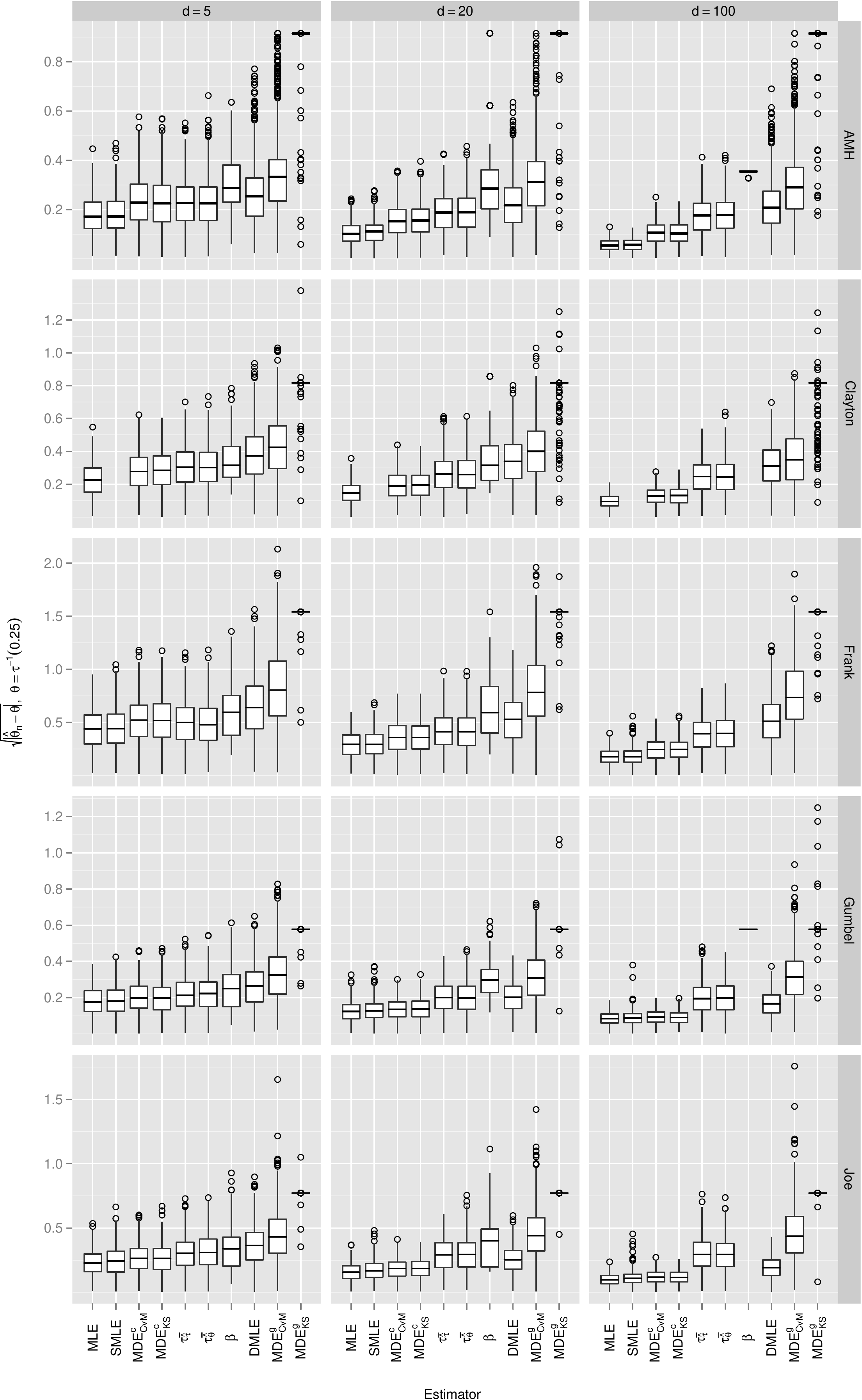}%
  \caption{Box plots of the square root of the absolute error (under known margins) obtained from $N=1000$ replications of sample size $n=100$ for Kendall's tau equal to 0.25.}
  \label{fig.err.tau.0.25}
\end{figure}
\begin{figure}[htbp]
  \centering
  \includegraphics[width=0.9\textwidth]{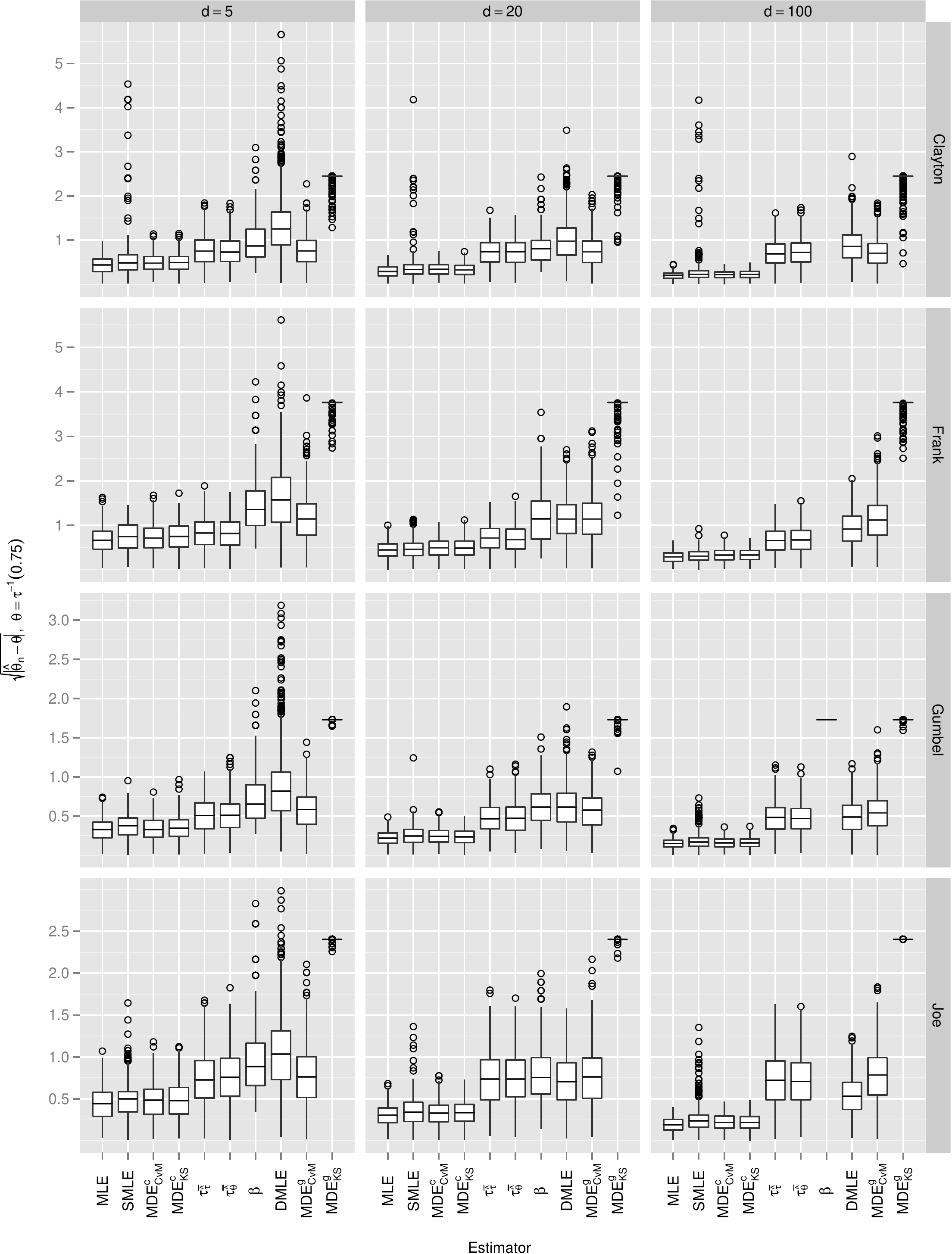}%
  \caption{Box plots of the square root of the absolute error (under known margins) obtained from $N=1000$ replications of sample size $n=100$ for Kendall's tau equal to 0.75.}
  \label{fig.err.tau.0.75}
\end{figure}

\subsection{Results under unknown margins}\label{sec.res.pobs}
Tables~\ref{tab.bias.pobs} and \ref{tab.rmse.pobs} in the appendix contain the bias (multiplied by 1000) and the RMSE (multiplied by 1000), respectively, for all investigated estimators based on pseudo-observations. Note that the run times are quite similar to those reported in Table \ref{tab.mut} and therefore omitted.

Figures \ref{fig.err.tau.0.25.pobs} and \ref{fig.err.tau.0.75.pobs} graphically display the corresponding square root of the absolute error via box plots obtained from the $N=1000$ replications. As for Figures \ref{fig.err.tau.0.25} and \ref{fig.err.tau.0.75}, we exclude methods which perform so poorly that their boxplots dominate the scale of values. Under pseudo-observations, these were the same methods as under known margins with the only exceptions being $\MDE{\Gamma}{KS}$ for $d\in\{20,100\}$ which are excluded and SMLE for Clayton which performed better under pseudo-observations and is thus included in the figures; see Section \ref{sec.num.smle} for an explanation.

The performance of the estimators based on pseudo-observations can be summarized as follows. Overall, the MLE still performs best, but the differences in absolute error are much less obvious. Furthermore, although a slight improvement of the performance of the MLE in larger dimensions $d$ is visible, the rate of improvement does not seem to be as large as under known margins. Concerning the estimators based on Kendall's tau and the minimum distance estimators, the former performs well for the case $\tau=0.25$, the latter performs well for the case $\tau=0.75$.
\begin{figure}[htbp]
  \centering
  \includegraphics[height=0.97\textheight]{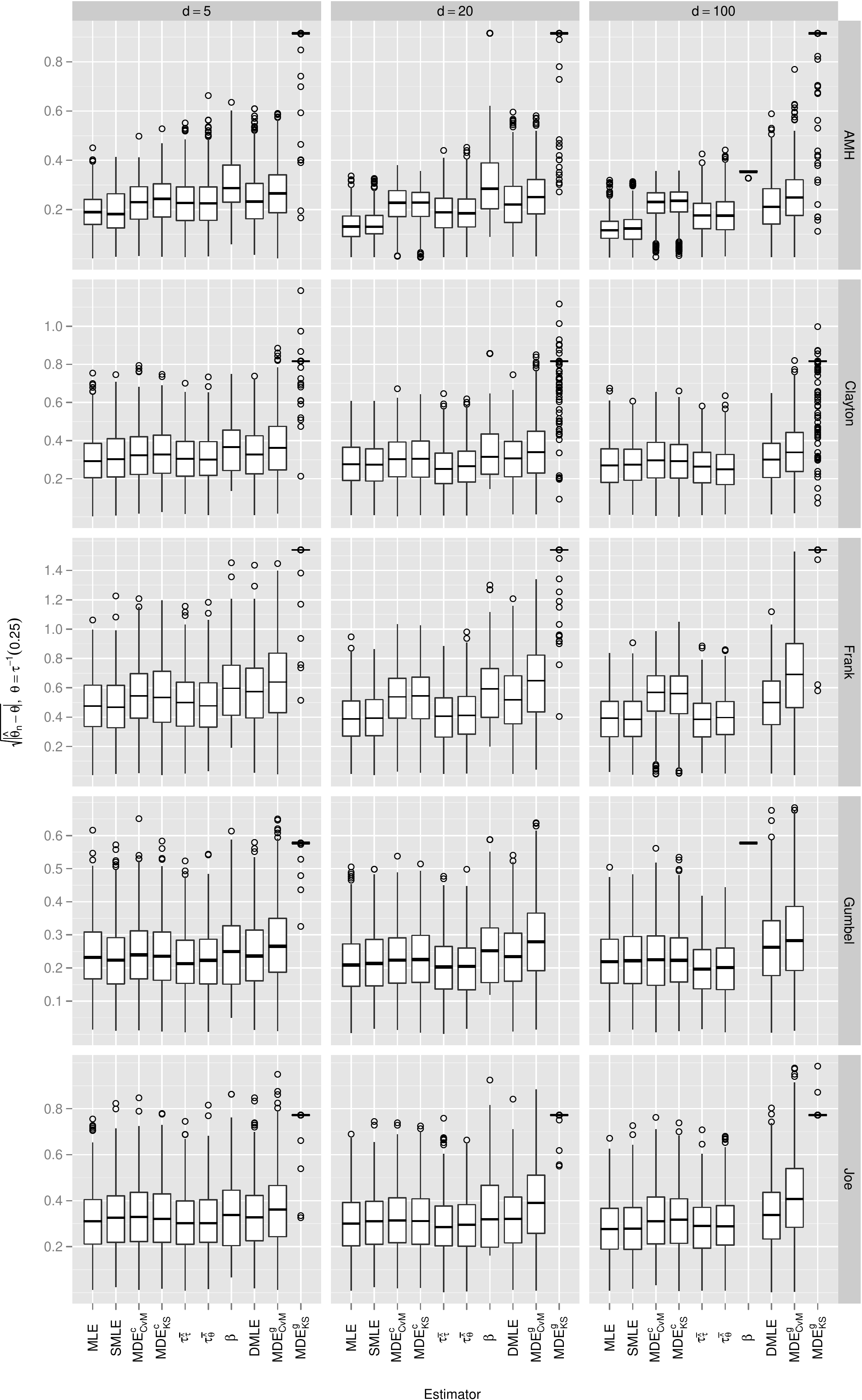}%
  \caption{Box plots of the square root of the absolute error (under unknown margins) obtained from $N=1000$ replications of sample size $n=100$ for Kendall's tau equal to 0.25.}
  \label{fig.err.tau.0.25.pobs}
\end{figure}
\begin{figure}[htbp]
  \centering
  \includegraphics[width=0.9\textwidth]{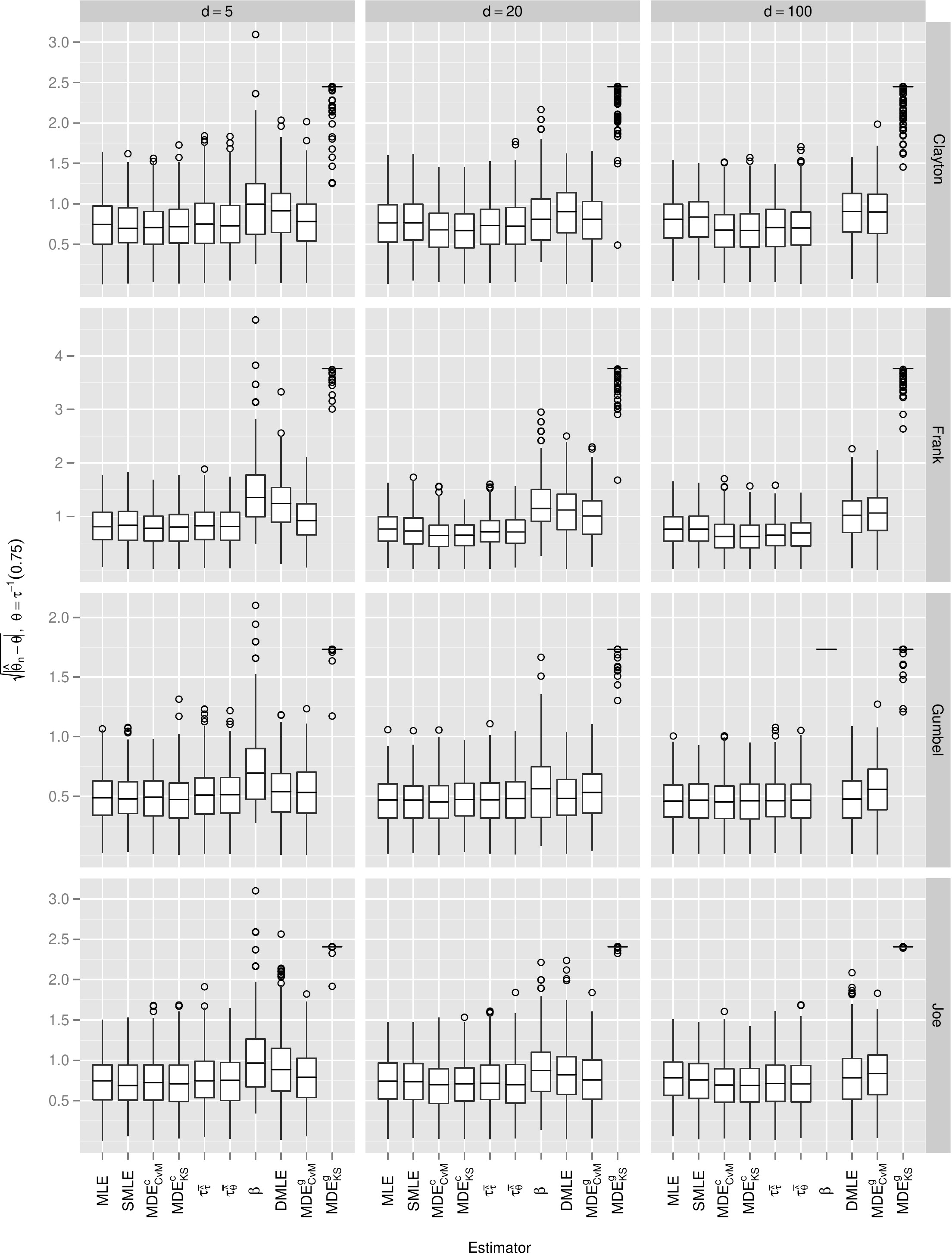}%
  \caption{Box plots of the square root of the absolute error (under unknown margins) obtained from $N=1000$ replications of sample size $n=100$ for Kendall's tau equal to 0.75.}
  \label{fig.err.tau.0.75.pobs}
\end{figure}

\section{Numerical issues and partial solutions}\label{sec.num}
In this section, we address some specific numerical problems we
encountered when working in high dimensions. These problems are not
trivial to solve and for some, no simple solution exists to date. We
included this section to emphasize that working in large dimensions is
much more affected by numerical issues. This is not merely a problem
of slow run times; it is also a huge problem for precision. As a general remark, let us stress that what is known about estimators in low dimensions does not always carry over to the high-dimensional case: Estimators that are fast in low dimensions may turn out to be too slow in large dimensions (although robust, the pairwise Kendall's tau estimators face this problem); estimators whose simple form suggest good performance in large dimensions may be highly prone to numerical errors (which is the case, for example, for Blomqvist's beta due to accessing the survival copula involved, a critical task in large dimensions).

\subsection{Minimum distance estimators}\label{sec.num.mde}
The minimum distance estimators were especially prone to the problem of a
flat objective function for the optimization for all but the
Ali-Mikhail-Haq family. Note that there, the
parameter $\theta$ runs in a bounded interval which is typically
advantageous for optimization.

To show the problem, we consider the Gumbel copula and pick out the Cram\'er-von Mises distances based on the mapping to a $\chi^2$ distribution (via the quantile function of the normal distribution) as described in Section \ref{sec.mde}. The left-hand side of Figure \ref{fig.mde} shows the objective function (the distance to be minimized) based on samples of size $n=100$ from Gumbel copulas in the dimensions $d\in\{5,20,100\}$ with parameter $\theta=4/3$ such that Kendall's tau equals 0.25. We choose the same (large) plotting interval as is chosen for the optimization in the simulation study in Section \ref{sec.sim}. As can be seen from this figure, the distance to be minimized becomes flat already for moderate parameter values. The optimization of this distance carried out in the simulation study is done via \R's \code{optimize}. It is indicated on the corresponding help page how this function proceeds. Based on the first two points in the optimization procedure, it is clear that the algorithm remains in the ``flat part'' of the distance function and thus returns wrong estimates.

The solution to this problem is simple and effective: By reparameterizing
the distance one can carry out the optimization without problems. To see
why, consider the right-hand side of Figure \ref{fig.mde} which shows
precisely the same distance as on the left-hand side of this figure, but
now plotted in $\alpha=1-1/\theta$. The advantage of this
reparameterization is that the objective function is now a function of the
bounded variable $\alpha\in(0,1]$.
\begin{figure}[htbp]
  \centering
  \includegraphics[width=0.48\textwidth]{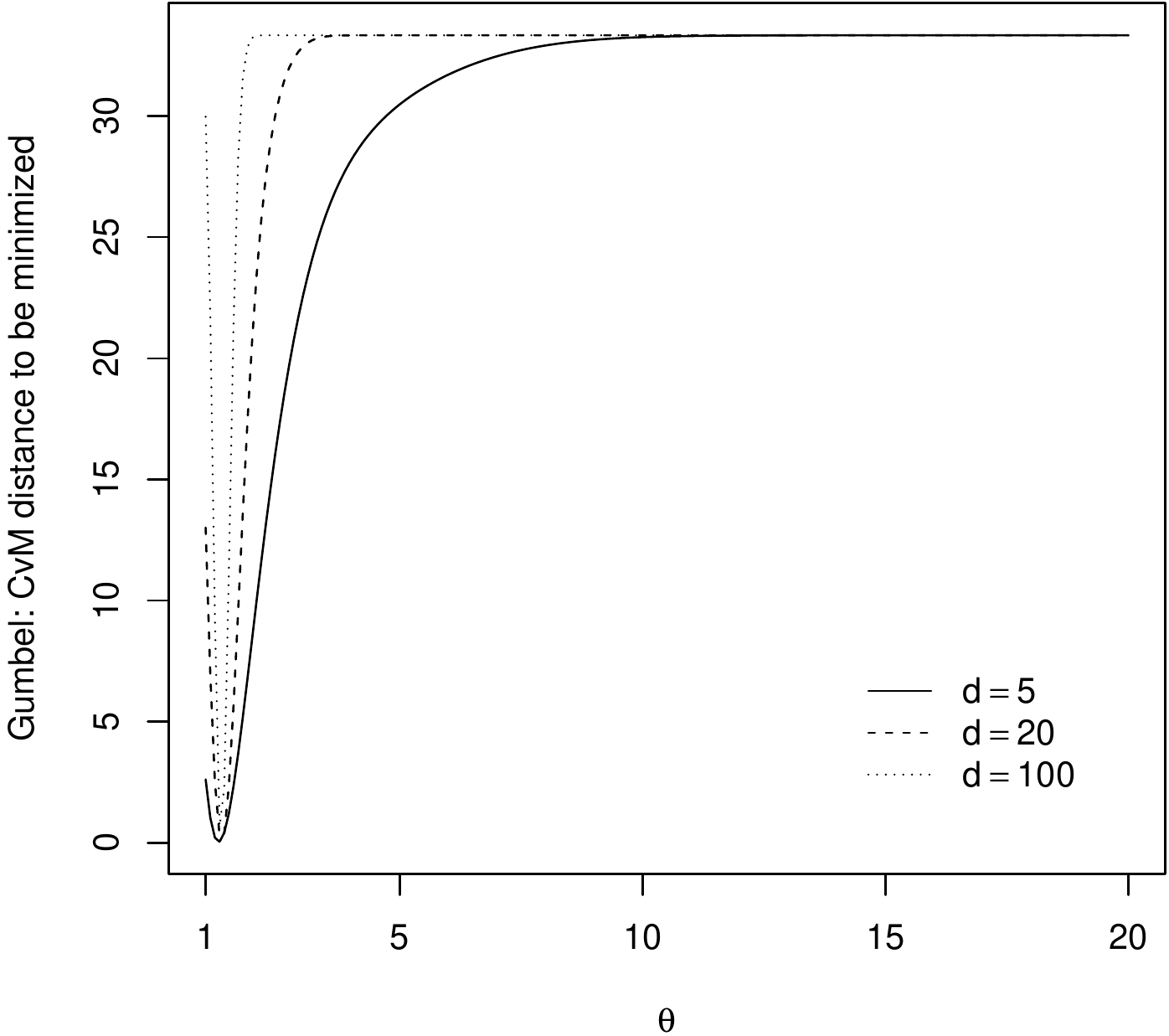}%
  \hfill
  \includegraphics[width=0.48\textwidth]{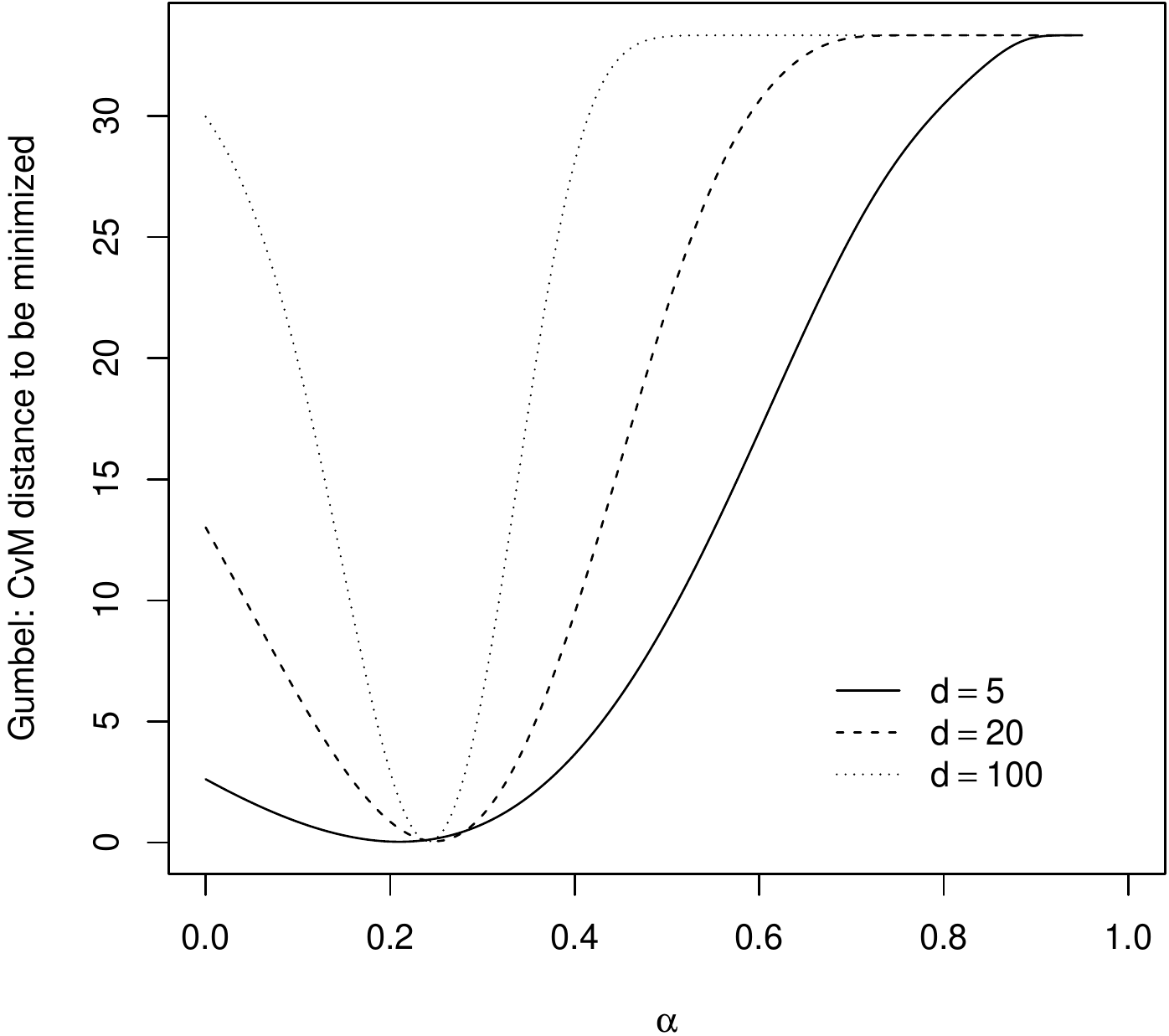}%
  \caption{Plot of the Cram\'er-von Mises distances (based on the mapping to a $\chi^2$ distribution) without (left) and with (right) reparameterization of the distance for the Gumbel copula with parameter $\theta=4/3$ (Kendall's tau equals 0.25) based on a sample of size $n=100$ in the dimensions indicated.}
  \label{fig.mde}
\end{figure}

Similar transformations turn out to be convenient for the families of Clayton, Frank, and Joe as well. For the latter, we use the same reparameterization as for Gumbel, for Clayton and Frank we use $\alpha=2\arctan(\theta)/\pi$.

\subsection{Simulated maximum-likelihood estimation}\label{sec.num.smle}
As can be seen from the results in Sections \ref{sec.res} and
\ref{sec.res.pobs}, the SMLE performs well except for Clayton's family
under known margins.  In this section we briefly investigate why. For this,
recall that the SMLE is based on the approximation (\ref{der.smle}). The
log-density approximated via this Monte Carlo method then involves
\begin{gather}
  \log((-1)^d\psi^{(d)}(t))\approx\log\biggl(\frac{1}{m}\sum_{k=1}^mV_k^d\exp(-V_kt)\biggr)=\log\biggl(\frac{1}{m}\sum_{k=1}^m\exp(b_k)\biggr),\label{log.der.smle}\\
  \text{where}\quad b_k=d\log(V_k)-V_kt.\notag
\end{gather}
For the SMLE to compute, we have to replace $t$ in (\ref{log.der.smle}) by
$\sum_{j=1}^d\psii(u_j)$. For simplicity, let us assume that all components
$u_j$ are equal to $u$, so we consider the vector
$\bm{u}=(u,\dots,u)\T\in[0,1]^d$. The corresponding value $t$ for
(\ref{log.der.smle}) is then $t=d\psii(u)=d(u^{-\theta}-1)$. Let us assume that $\theta=2$, that is, the corresponding value of Kendall's tau is 0.5. If $u$ is small, then $t$ becomes large. The problem is now that the exponents $b_k$ become quite small. Indeed, they are so small (depending on $t$) that the $\exp(b_k)$ become zero in computer arithmetic for many (again depending on $t$) of the $m=10\,000$ sampled $V_k$'s. These zeros significantly affect the approximation in (\ref{log.der.smle}).

To give an example, let $\theta=2$, $d=5$, draw i.i.d.\ $V_k\sim\Gamma(1/\theta,1)$, $k\in\{1,\dots,m\}$, for $m=10\,000$ (using \code{set.seed(1)}), and compute the approximation in (\ref{log.der.smle}) at $t\in\{5\cdot 10^{16}, 5\cdot 10^{12}, 5\cdot 10^8, 15\}$. The left-hand side (correct values) for these values of $t$ are (roughly) -208.09, -157.44, -106.78, and -11.86, whereas the right-hand side gives -\code{Inf}, -622.62, -124.41, and -11.86. As one can see, for $t=15$ both values agree, for $t=5\cdot 10^8$, the approximation is already quite far away from the corresponding true value. For large $t$ this problem becomes more severe, with the extreme case being such a large $t$ that $\exp(b_k)$ is zero in computer arithmetic for all $V_k$'s. This implies that $\log(0)=-\code{Inf}$ in (\ref{log.der.smle}). Note that this could be avoided by using an intelligent logarithm as given in Lemma \ref{lem.lsum.lssum} \ref{lem.lsum} below. However, this does not solve the problem of a poor approximation to $\log((-1)^d\psi^{(d)}(t))$ (see also Figure \ref{fig.smle} below), the problem being that all summands are zero, except the one being $\exp(b_k-b_{\text{max}})=\exp(0)=1$.

Figure \ref{fig.smle} shows, in log-log scale, the relative error
of the approximation (\ref{log.der.smle}) for $\bm{u}=(u,\dots,u)\T$ as a
function of $u$, based on $m=10\,000$, $d=5$, and $\theta=2$.  As is
clearly visible, the relative error of the approximation becomes much
larger for smaller values of $u$. Since the Clayton copula has lower tail
dependence, there is indeed a positive probability of obtaining random
vectors with simultaneously small components. These (and only these)
samples affect the likelihood approximation and lead to wrong SMLEs.
Note that this problem vanishes for the SMLE based on pseudo-observations, see Figures \ref{fig.err.tau.0.25.pobs} and \ref{fig.err.tau.0.75.pobs}, since each $\hat{u}_{ij}\in\{1/(n+1),\dots,n/(n+1)\}$ and thus the $\hat{u}_{ij}$'s are naturally bounded from below by $1/(n+1)$.
\begin{figure}[htbp]
  \centering
  \includegraphics[width=0.5\textwidth]{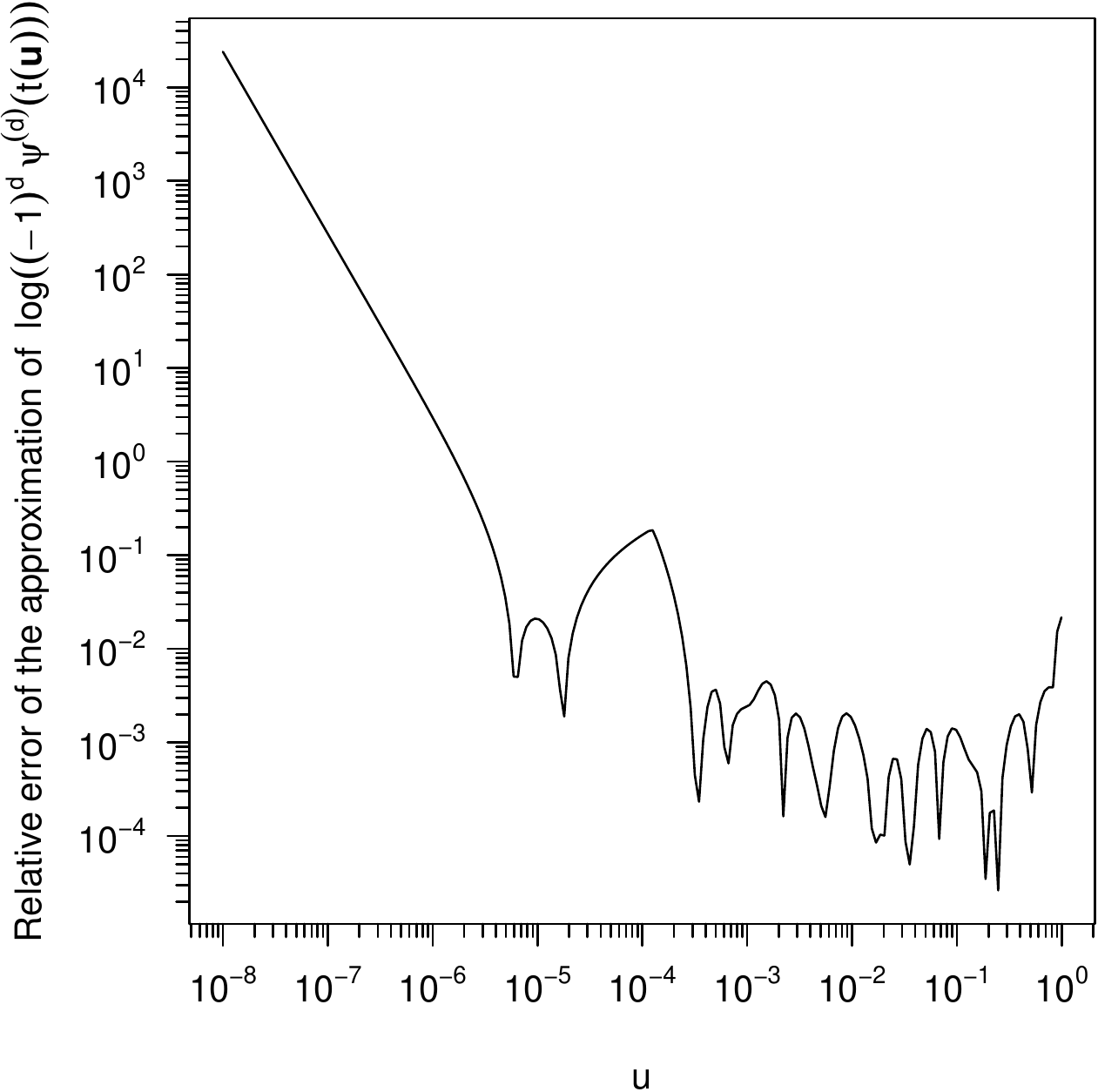}%
  \caption{Relative error as a function of $u$ based on $m=10\,000$, $d=5$, and $\theta=2$.}
  \label{fig.smle}
\end{figure}

\subsection{Gumbel's and Joe's polynomial}\label{sec.num.GJ}
For computing the log-likelihood for the Archimedean Gumbel or Joe copula, we need an efficient way of evaluating the logarithm of the density $c_\theta(\bm{u})$ as given in Section \ref{sec.mle}, Parts \ref{sec.mle.G} and \ref{sec.mle.J}, respectively. The challenge is to evaluate the logarithm of the polynomials involved. For this the following auxiliary results are essential. Their proofs are straightforward and thus omitted.
\begin{lemma}\label{lem.lsum.lssum}
  \begin{enumerate}
  \item\label{lem.lsum} Let $x_i\ge0$, $i\in\{1,\dots,n\}$, such that $\sum_{i=1}^nx_i>0$. Furthermore, let $b_i=\log x_i$, $i\in\{1,\dots,n\}$, with $\log 0=-\infty$, and let $b_{\text{max}}=\max_{1\le i\le n}b_i$. Then
    \begin{align}
      \log\sum_{i=1}^nx_i=b_{\text{max}}+\log\sum_{i=1}^n\exp(b_i-b_{\text{max}}).\label{lsum}
    \end{align}
  \item\label{lem.lssum} Let $x_i\in\IR$, $i\in\{1,\dots,n\}$, such that $\sum_{i=1}^nx_i>0$. Furthermore, let $s_i=\sign x_i$, $b_i=\log\lvert x_i\rvert$, $i\in\{1,\dots,n\}$, with $\log 0=-\infty$ and let $b_{\text{max}}=\max_{1\le i\le n}b_i$. Then
    \begin{align}
      \log\sum_{i=1}^nx_i=b_{\text{max}}+\log\biggl(\sum_{i=1\atop i:\,s_i=1}^n\exp(b_{(i)}-b_{\text{max}})-\sum_{i=1\atop i:\,s_i=-1}^n\exp(b_{(i)}-b_{\text{max}})\biggr),\label{lssum}
    \end{align}
    where $b_{(i)}$ denotes the $i$th smallest value of $b_i$, $i\in\{1,\dots,n\}$.
  \end{enumerate}
\end{lemma}
The ideas behind Lemma \ref{lem.lsum.lssum} \ref{lem.lsum} and \ref{lem.lssum} are implemented in the (non-exported) functions \code{lsum} and \code{lssum} in the \R{} package \pkg{copula}.

Although mathematically straightforward, Lemma \ref{lem.lsum.lssum} has an important consequence for evaluating logarithms of polynomials such as $\var{P}{G}{d,\alpha}$ for Gumbel's or $\var{P}{J}{d,\alpha}$ for Joe's density. Depending on the evaluation point, it might happen that the value of the polynomial is not representable in computer arithmetic and thus one cannot first compute the value of the polynomial and take the logarithm afterwards. Instead, Formula (\ref{lsum}) suggests a ``intelligent'' (numerically stable) logarithm to compute such polynomials (or sums). By taking out the maximum of the $b_i$, it is guaranteed that the exponentials which are summed up are all in $[0,1]$ and thus the sum takes on values in $[0,n]$, representable in computer arithmetic. This trick solves the numerical issues for computing $\var{P}{J}{d,\alpha}$ and thus for computing the log-likelihood of a Joe copula. The evaluation of $\var{P}{J}{d,\alpha}$ is implemented as (non-exported) function \code{polyJ} in the \R{} package \pkg{copula}. It is called with default method \code{log.poly} implementing the trick described above when evaluating the density of a Joe copula via the slot \code{dacopula}; two other, less efficient methods are also available, one of which is a straightforward polynomial evaluation (\code{poly}).

Formula (\ref{lssum}) takes the above idea of an intelligent logarithm a step further, by dealing with possibly negative summands. The summands in each sum are ordered in increasing order to prevent cancellation. This formula is helpful in computing $\var{P}{G}{d,\alpha}$. However, the situation turns out to be more challenging for Gumbel's family. All in all, several different methods for the evaluation of $\var{P}{G}{d,\alpha}$ were implemented. They are based on the following results about $\var{P}{G}{d,\alpha}$ and described below, where here and in the following, $\alpha=1/\theta \in(0,1]$.
\begin{lemma}\label{lem.polyG}
  Let
  \begin{align}
    \var{P}{G}{d,\alpha}(x)=\sum_{k=1}^d\var{a}{G}{dk}(\alpha)x^k\label{polyG}
  \end{align}
  for $\alpha\in(0,1]$, where
  \begin{align}
  \var{a}{G}{dk}(\alpha)&=(-1)^{d-k}\sum_{j=k}^d\alpha^{j}s(d,j)S(j,k)\label{polyG.coeff.stir}\\
&=\frac{d!}{k!}\sum_{j=1}^k\binom{k}{j}\binom{\alpha j}{d}(-1)^{d-j}\label{polyG.coeff.binom},\ k\in\{1,\dots,d\}.
  \end{align}
  Then \begin{enumerate}
  \item\label{dsumSibuya} $\var{a}{G}{dk}(\alpha)=\var{p}{J}{01}(d;k)d!/k!$ for all $\alpha\in(0,1]$, where $\var{p}{J}{01}(d;k)>0$ denotes a probability mass function in $d\in\{k,k+1,\dots\}$;
  \item\label{polyG.repr} $\var{P}{G}{d,\alpha}$ allows for the following representations:
    \begin{align}
      \var{P}{G}{d,\alpha}(x)&=(-1)^{d-1}x\sum_{j=1}^d\biggl(s(d,j)\sum_{k=0}^{j-1}S(j,k+1)(-x)^k\biggr)\alpha^j\label{stirling}\\
&=(-1)^{d-1}\alpha x\sum_{j=0}^{d-1}\biggl(s(d,j+1)\sum_{k=0}^{j-1}S(j,k+1)(-x)^{k}\biggr)\alpha^j\label{stirling.horner}\\
      &=\sum_{j=1}^{d}s_j\exp\bigl(\log\lvert(\alpha j)_d\rvert+j\log x+x-\log(j!)+\log F^{\text{Poi}(x)}(d-j)\bigr),\label{pois}
    \end{align}
    where, for all $j\in\{1,\dots,d\}$,
    \begin{align}
      s_j=\begin{cases}\label{signFF}
        (-1)^{j-\lceil\alpha j\rceil},&\alpha j\notin\IN\ \text{or}\ (\alpha=1,\ j=d),\\
        0,&\text{otherwise},
      \end{cases}
    \end{align}
    and where $F^{\text{Poi}(x)}(\cdot\,)$ denotes the distribution function of a Poisson distribution with parameter $x$.
  \end{enumerate}
\end{lemma}
\begin{proof}
  Part \ref{dsumSibuya} of Lemma \ref{lem.polyG} follows from \cite[p.\ 99]{hofert2010c} (the probability mass function $\var{p}{J}{01}(d;k)>0$ corresponds to the distribution function $F_{01}(\cdot\,;k)$ whose Laplace-Stieltjes transform is the generator $\psi_{01}(\cdot\,;k)$ appearing in a nested Joe copula). In particular, this equality implies that the coefficients $\var{a}{G}{dk}(\alpha)$ of $\var{P}{G}{d,\alpha}$ are positive. This allows one to apply (\ref{lsum}) with $b_k=\log\var{a}{G}{dk}(\alpha)+k\log x$, $k\in\{1,\dots,d\}$ ($d=n$), to compute the logarithm of the polynomial $\var{a}{G}{dk}(\alpha)$ at $x$. Concerning Part \ref{polyG.repr}, Equations (\ref{stirling}) and (\ref{stirling.horner}) directly follow from interchanging the order of summation of (\ref{polyG}) combined with (\ref{polyG.coeff.stir}). For Equation (\ref{pois}), note that interchanging the order of summation of (\ref{polyG}) combined with (\ref{polyG.coeff.binom}) and rewriting the generalized binomial coefficient $\binom{\alpha j}{d}$ as $(\alpha j)_d/j!$ leads to
  \begin{align*}
    \var{P}{G}{d,\alpha}(x)=\sum_{j=1}^{d}(\alpha j)_d(-1)^{d-j}\frac{x^j\exp(x)}{j!}\sum_{k=0}^{d-j}\frac{x^k}{k!}\exp(-x).
  \end{align*}
  Interpreting the second sum as $F^{\text{Poi}(x)}(d-j)$, pulling out the signs $s_j=\sign((\alpha j)_d(-1)^{d-j})$, and bringing in an $\exp(\log(\cdot))$ leads to the result as stated. Concerning the formula for $s_j$, note that
  \begin{align}
    s_j&=\sign((\alpha j)_d(-1)^{d-j})=(-1)^{d-j}\sign\prod_{l=1}^{d-1}(\alpha j-l)=(-1)^{d-j}\prod_{l=1}^{d-1}\sign(\alpha j-l)\notag\\
&=(-1)^{d-j}\prod_{l=\lceil\alpha j\rceil}^{d-1}\sign(\alpha j-l).\label{signFF.repr}
  \end{align}
  First assume $\alpha j\notin\IN$ and $\lceil\alpha j\rceil\le d-1$. In this case $s_j=(-1)^{d-j}(-1)^{d-1-\lceil\alpha j\rceil+1}=(-1)^{j-\lceil\alpha j\rceil}$. This is also true if $\lceil\alpha j\rceil=d$ since then $j$ has to be equal to $d$. Now consider $\alpha=1$ and $j=d$. In this case, it is easily seen from (\ref{signFF.repr}) that $s_j=1$ which is also true for (\ref{signFF}). Finally, consider the second case in (\ref{signFF}). It implies that $\alpha j\in\IN$ and thus the first factor in (\ref{signFF.repr}) being zero due to $l=\lceil\alpha j\rceil=\alpha j$, so $s_j=0$. Note the interesting fact that the formula for $s_j$ is independent of $d$.
\end{proof}

The results presented in Lemma \ref{lem.polyG} lead to the following different methods for computing the logarithm of $\var{P}{G}{d,\alpha}$, see the (non-exported) function \code{polyG} of the \R{} package \pkg{copula}:
\begin{itemize}
\item\code{pois}, \code{pois.direct}: These methods are based on Representation (\ref{pois}), where \code{pois} applies the intelligent logarithm (\ref{lssum}) to evaluate the sum and \code{pois.direct} computes the sum directly as given in (\ref{pois});
\item\code{stirling}, \code{stirling.horner}: Method \code{stirling} evaluates
  Representation (\ref{stirling}) directly, where the polynomial
  $\sum_{k=0}^{j-1}S(j,k+1)(-x)^k$ in $-x$ is computed via Horner's
  scheme. Method \code{stirling.horner} is based on Representation
  (\ref{stirling.horner}), where Horner's scheme is applied to compute the
  polynomial in $\alpha$ with coefficients $s(d,j+1)\sum_{k=0}^{j-1}S(j,$ $k+1)(-x)^{k}$ which, as before, are evaluated with Horner's scheme.
\item\code{sort}, \code{horner}, \code{direct}, and \code{dsSib.*}: These methods all $\var{P}{G}{d,\alpha}$ with the intelligent logarithm (\ref{lsum}) based on the logarithms of the coefficients $\var{a}{G}{dk}(\alpha)$ (note that the coefficients $\var{a}{G}{dk}(\alpha)$ of  are all positive). The logarithmic coefficients can be obtained in different ways: \code{sort} computes them via (\ref{lssum}); \code{horner} via Horner's scheme based on interpreting (\ref{polyG.coeff.stir}) as a polynomial in $-\alpha$; \code{direct} by directly computing the sum as given in (\ref{polyG.coeff.stir}); and \code{dsSib.*} by various different methods described on the help page of the function \code{dsumSibuya} (for example, \code{dsSib.log} uses (\ref{polyG.coeff.binom}) together with the intelligent logarithm as given in (\ref{lssum})).
\end{itemize}
Additionally, a method \code{default} is implemented which consists of a careful combination of the above methods based on numerical experiments. Note that all methods involved work with $\log x$ instead of $x$.
For this reason, \code{polyG} requires as argument $\log x$ rather than $x$.

Finally, let us mention that the problem of evaluating sums of type
\begin{align}\label{sumBinom}
  \sum_{k=0}^n {n\choose k}(-1)^k f_k
\end{align}
for sequences $(f_k)_k$ has a long history (note that (\ref{polyG.coeff.binom}) falls under this setup).
They can be interpreted as forward differences and are known to be
numerically challenging. Approximate (asymptotic for $n\to\infty$) formulas
may be obtained by using methods from complex analysis; see, for example,
\cite{flajoletsedgewick1995}, and have been important, e.g., for estimating
the complexity of computer algorithms.

However, these asymptotic formulas are not very accurate for finite $n$
(note that in our case, $n=d$, the data dimension) and the only known
way to accurately compute them, seems high precision arithmetic.
See \code{sumBinomMpfr()} in \R{} package \pkg{Rmpfr}, and its
documentation for simple examples such as $f_k = f(k) = \sqrt{k}$.

\subsection{Kendall's tau for Ali-Mikhail-Haq copulas}
In Table~\ref{tab.gen}, the population version of Kendall's tau
for the Ali-Mikhail-Haq (A) family, as function of the parameter $\theta$,
is
\begin{align}
  \label{eq:tau.AMH}
  \tau_{\text{A}}(\theta) = 1-2(\theta+(1-\theta)^2\log(1-\theta))/(3\theta^2).
\end{align}
When computing it for the Kendall's tau estimator (see Section~\ref{sec.tau}), however,
the simple formula (\ref{eq:tau.AMH}) is not sufficient, notably not for
small $\theta$, see the left plot in Figure~\ref{fig.AMHtau}:  Replacing
$\log(1-\theta)$ by its numerical accurate $\mathrm{log1p}(-\theta)$ helps
down to around $\theta\approx 10^{-7}$, but then that formula breaks down
as well, and indeed our \code{tauAMH()} (package \pkg{copula}), uses
parts of the Taylor series $\tau_{\text{A}}(\theta) = \frac{2}{9}\theta(1 + \theta(\frac 1 4 +
    \frac{\theta}{10}(1 + \theta(\frac 1 2 + \theta \frac 2 7)))) +
    O(\theta^6)$%
    \footnote{replacing $\log(1-\theta)$ by its Taylor
      expansion $- \sum_{k=1}^\infty \theta^k/k$ in (\ref{eq:tau.AMH})
      results in the expansion $\tau_{\text{A}}(\theta) = \frac 2 9
      \theta\sum_{k=0}^\infty \frac{6}{(k+1)(k+2)(k+3)}\,\theta^k$},
    as soon as $\theta \le 10^{-2}$.
\begin{figure}[htbp]
  \centering
  \includegraphics[width=0.485\textwidth]{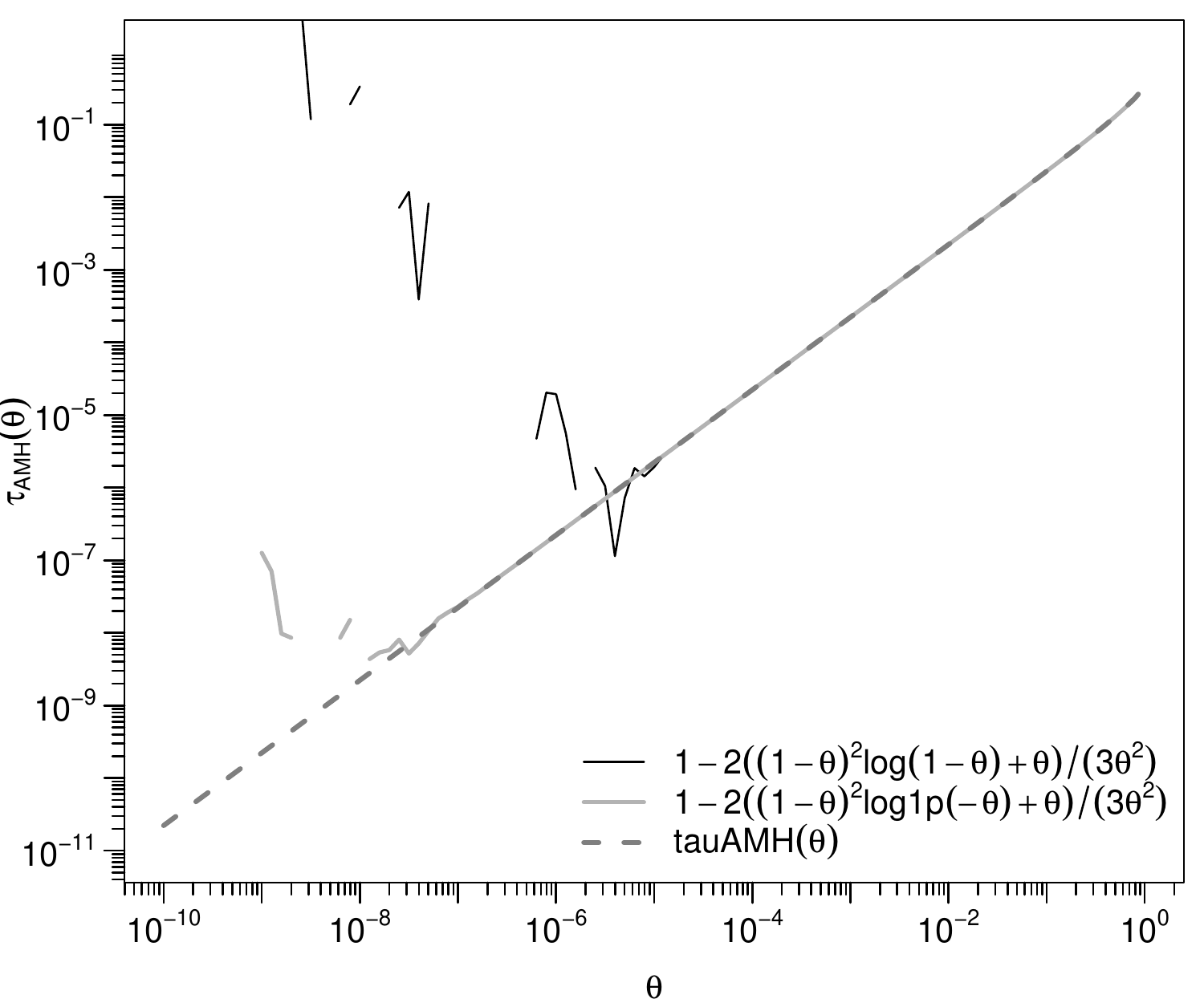}%
  \hfill
  \includegraphics[width=0.485\textwidth]{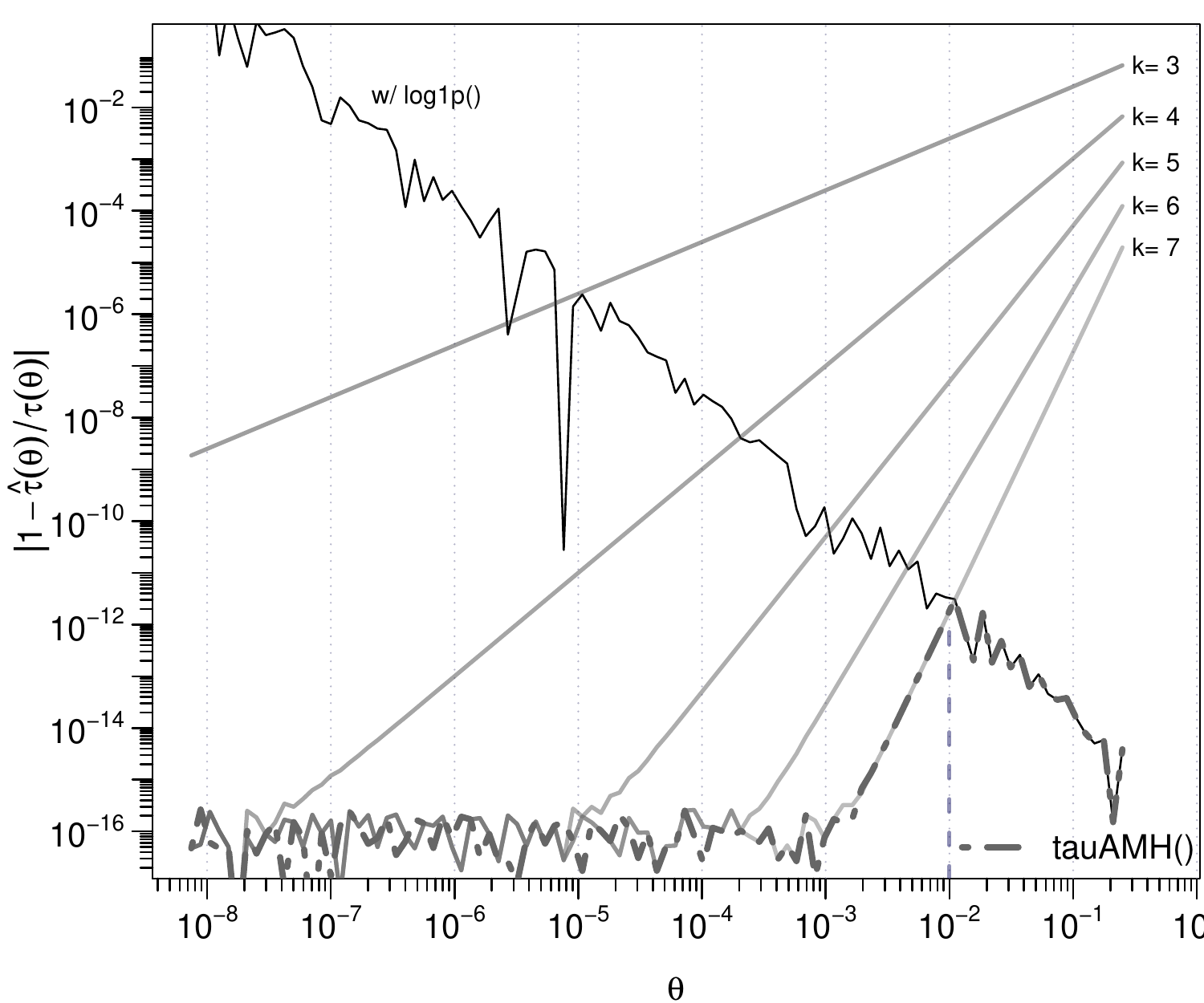}
  \caption{left: $\tau_{\text{A}}(\theta)$: direct form (breaking down for $\theta<10^{-5}$), using \code{log1p()} (okay down to $\theta\approx 10^{-8}$), correct approximation as provided by \code{tauAMH()}; right: relative errors of \code{log1p()}, Taylor approximations, and our hybrid \code{tauAMH()}.}
  \label{fig.AMHtau}
\end{figure}

\subsection{log1mexp}\label{sec:log1mexp}
There are several situations, such as the one addressed in Section \ref{sec.num.dfrank}, where an accurate computation of
\begin{align}
  \label{eq:log1mexp}
  f(a) = \mathrm{log1mexp}(a) = \log(1 - \exp(-a)),\quad a \ge 0,
\end{align}
is required. Note that this is numerically challenging in both situations, when $a \downarrow
0$ (hence $\exp(-a) \uparrow 1$ and cancellation of two almost equal terms in $1
- \exp(-a)$), and when $a\uparrow\infty$, as $\exp(-a) \downarrow 0$ and in $1 -
\exp(-a)$, almost all accuracy of $\exp(-a)$ is lost when it is less than
around $10^{-15}$.
Now, for the first case, we can make use of the \R{} and C library function
\code{expm1(x)} which computes $\exp(x) - 1$ accurately also for very
small $x$, and for the second case, use the \R{} and C library function
\code{log1p(x)} which computes $\log(1 + x)$ accurately also for very
small $x$.  Our package \pkg{copula} provides the function
\code{log1mexp()} which adapts to these two cases, in a sense, optimally
by using a cutoff of $a = \log 2$; see \cite{maechler2012}.

\subsection{The density of the diagonal of Frank copulas}\label{sec.num.dfrank}
Computing the DMLE for Frank's copula family is one situation where the accurate computation of \eqref{eq:log1mexp} is crucial. To compute the density $\delta^\prime_{\theta}$ of the diagonal \eqref{diagdensity} for Frank copulas with generator $\psi_{\theta}(t)=-\log(1-(1-\exp(-\theta))e^{-t})/\theta$, the functions
\begin{gather*}
  -\psi^\prime_{\theta}(t)=\frac{(1-e^{-\theta})\exp(-t)}{\theta(1-(1-e^{-\theta})\exp(-t))},\\
 \psiis{\theta}(u)=-\log\biggl(\frac{\exp(-u\theta)-1}{e^{-\theta}-1}\biggr),\quad\text{and} -(\psiis{\theta})^\prime(u)=\frac{\theta}{\exp(\theta u)-1}
\end{gather*}
are involved. Numerical issues in computing $\delta^\prime_{\theta}(u)$ arise for large $\theta$ and $u$ close to 1. It is known that numerically, the computation of $e^x-1$ suffers from cancellation when $0<x\ll 1$. The first suspect is thus $\psiis{\theta}$ which involves terms of this type. The left-hand side of Figure \ref{fig.dFrank.psiInv} displays $\log\delta^\prime_{\theta}(u)$ for $\theta=38$ and $d=2$ for two different versions of computing $\psiis{\theta}$: \code{psiInv.0} uses only the \R{} functions \code{log()} and \code{exp()}, whereas \code{psiInv.1} uses \code{log()} and \code{expm1()}. Either way, numerical issues appear due to the cancellation in the division of terms of type $e^x-1$ when computing $\psiis{\theta}$. By rewriting $\psiis{\theta}$ as
\begin{align*}
  \psiis{\theta}(u)=-\log\biggl(1-\frac{\exp(-u\theta)-e^{-\theta}}{1-e^{-\theta}}\biggr)
\end{align*}
we can use \R{}'s function \code{log1p(.)} to accurately compute $\log(1+\cdot)$ and thus $\psiis{\theta}$ via \code{-log1p((exp(-u*theta)-exp(-theta))/expm1(-theta))} which we denote by \code{psiInv.2}. The right-hand side of Figure \ref{fig.dFrank.psiInv} displays the effect of using \code{psiInv.2} in comparison to \code{psiInv.0} and \code{psiInv.1}.

\begin{figure}[htbp]
  \centering
  \includegraphics[width=0.485\textwidth]{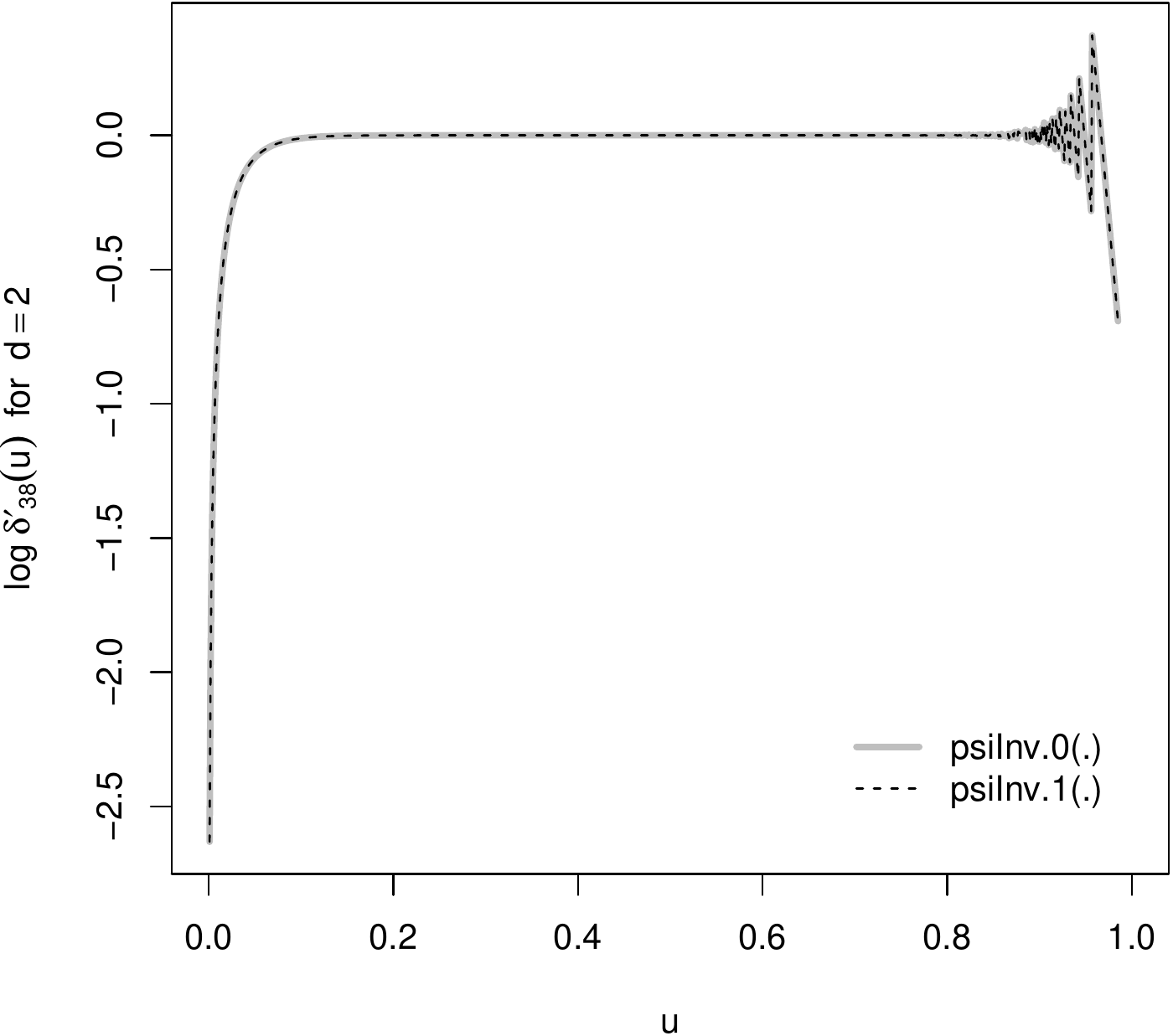}%
  \hfill
  \includegraphics[width=0.485\textwidth]{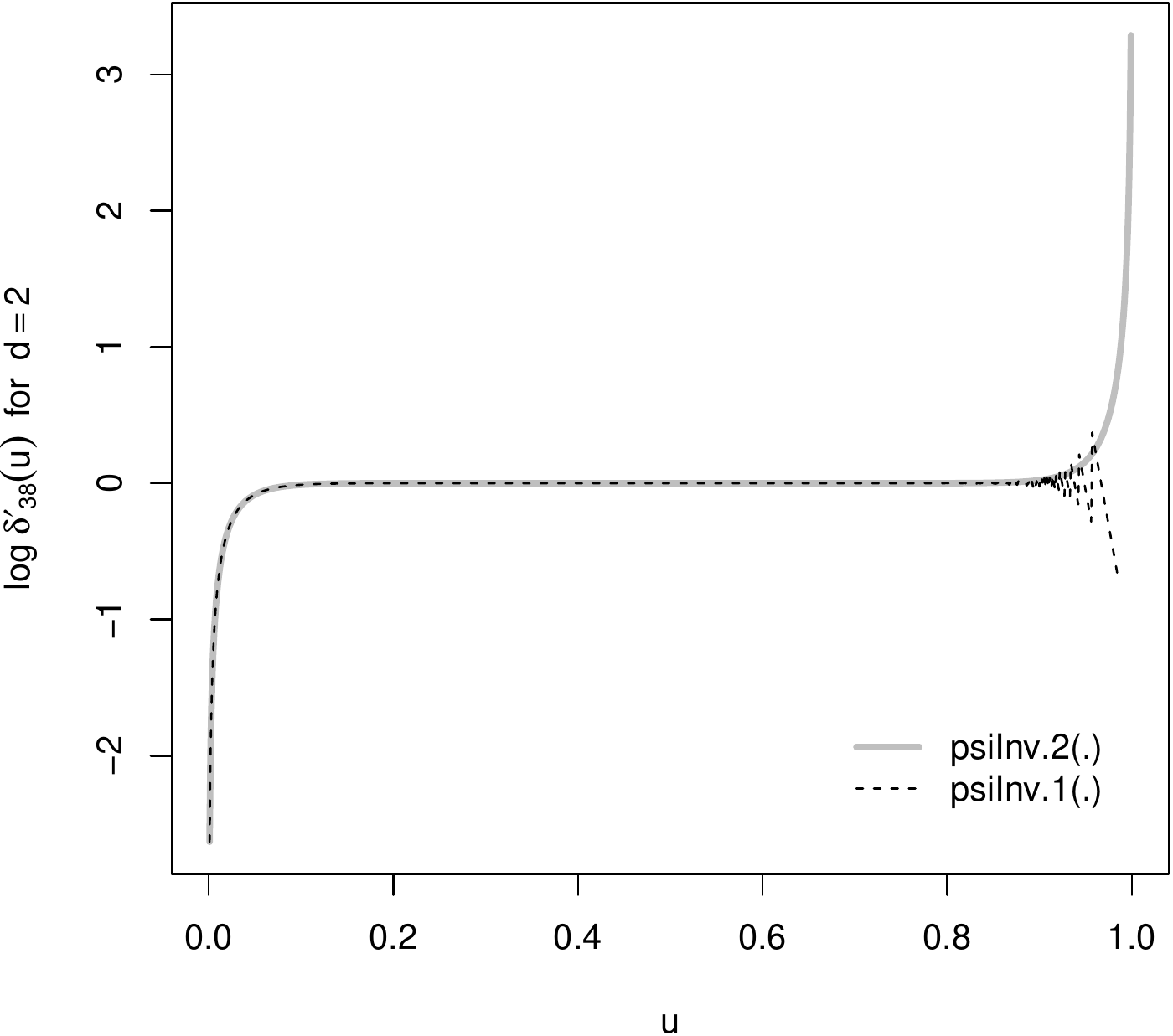}
  \caption{$\log\delta^\prime_{\theta}(u)$ for $\theta=38$ and $d=2$ for computing $\psiis{\theta}$ via \R{}'s \code{log()} and \code{exp()} functions (\code{psiInv.0}), a version with \code{log()} and \code{expm1()} (\code{psiInv.1}), and a version using \code{log1p()} and \code{expm1()} (\code{psiInv.2}).}
  \label{fig.dFrank.psiInv}
\end{figure}
Although this already looks promising, it is still not possible to compute the negative log-likelihood for the DMLE of Frank's copula family for a large range of parameters $\theta$ as one would like to do for the optimization. The left-hand side of Figure \ref{fig.dFrank.psiDabs} shows the negative log-likelihood based on the diagonal of a five-dimensional (so rather low-dimensional) Frank copula, where computations are done in double precision and high-precision arithmetic with different significant bits (this was done with the \R{} package \pkg{Rmpfr}). As it turns out, the problem is the evaluation of $-\psi^\prime_{\theta}(t)$ for small $t$ (equivalently, $t=\psiis{\theta}(u)$ for large $\theta$ and $u$ close to 1 as before). The solution is to rewrite the logarithm of $-\psi^\prime_{\theta}(t)$ via
\begin{align*}
  \log(-\psi^\prime_{\theta}(t))&=\log(1-e^{-\theta})-t-\log(\theta)-\log(1-(1-e^{-\theta})\exp(-t))\\
  &=\mathrm{log1mexp}(\theta)-t-\log(\theta)-\mathrm{log1mexp}\bigl(-\log((1-e^{-\theta})\exp(-t))\bigr)\\
  &=\mathrm{log1mexp}(\theta)-t-\log(\theta)-\mathrm{log1mexp}\bigl(t-\log(1-e^{-\theta})\bigr)\\
  &=w-\log(\theta)-\mathrm{log1mexp}(-w),
\end{align*}
where $w=\mathrm{log1mexp}(\theta)-t$. By computing $\mathrm{log1mexp}$ via
\code{log1mexp()} as described in Section~\ref{sec:log1mexp}, one can then
accurately compute the negative log-likelihood for the  DMLE for Frank's
copula family; see the right-hand side of
Figure~\ref{fig.dFrank.psiDabs}.  For more details we refer the interested
reader to \cite{Maechler2011}\footnote{As this is a vignette
  of \R{} package \pkg{copula}, all its figures are completely reproducible
  via \R{} code in the file \texttt{Frank-Rmpfr.Rnw} which is part of the
  package source.}.
\begin{figure}[htbp]
  \centering
  \includegraphics[width=0.485\textwidth]{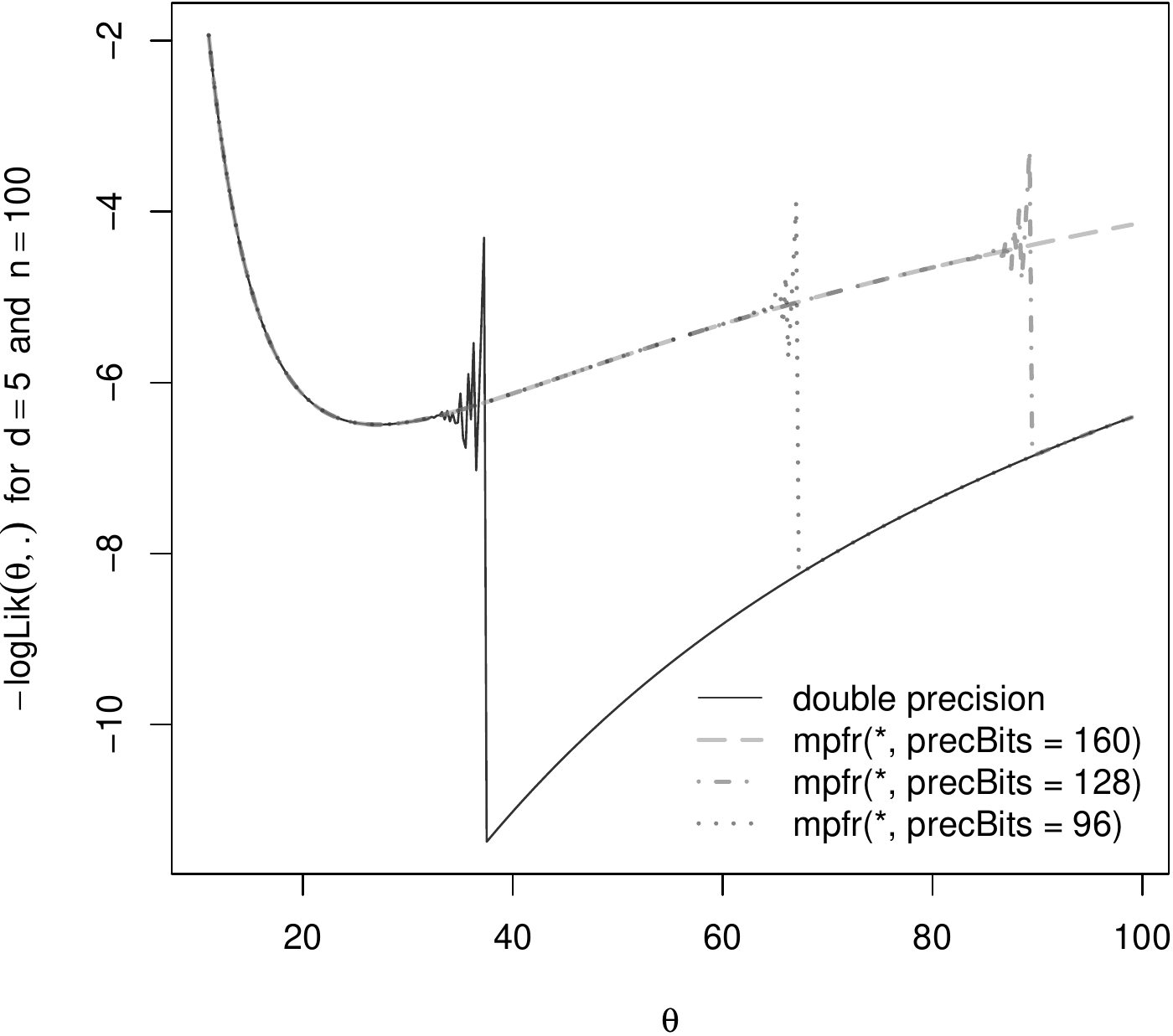}%
  \hfill
  \includegraphics[width=0.485\textwidth]{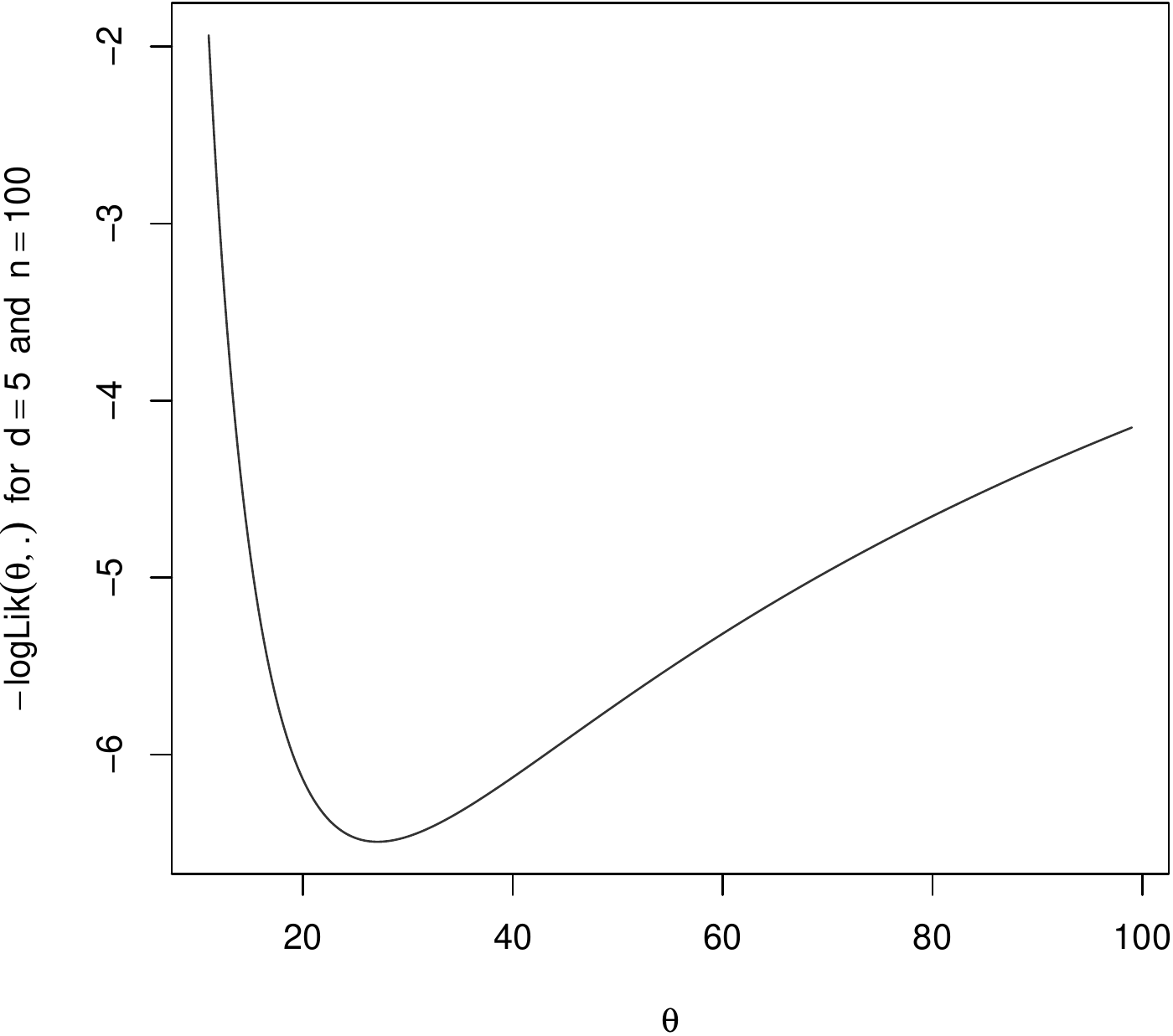}
  \caption{Negative log-likelihood based on the diagonal of a five-dimensional Frank copula (sample size $n=100$) with various kinds of precision (left) and an accurate version in double precision based on $\mathrm{log1mexp}$ (right).}
  \label{fig.dFrank.psiDabs}
\end{figure}

\section{Conclusion}\label{sec.con}
We introduced and compared different parametric estimators for Archimedean copula families with focus on large dimensions (up to $d=100$). In particular, estimators based on Kendall's tau, Blomqvist's beta, minimum distance estimators, the maximum-likelihood estimator, a simulated maximum-likelihood estimator, and a maximum-likelihood estimator based on the copula diagonal were investigated both under known and unknown margins (pseudo-observations). Several of these estimation methods were newly introduced and investigated in this context.

Under known margins, the best performance according to precision was shown by the maximum-likelihood estimator. To our surprise, the maximum-likelihood estimator also performed well according to numerical stability (being of similar numerical stability as the pairwise Kendall's tau estimators) and run time (being only outperformed by the diagonal maximum likelihood estimator). Under unknown margins, the MLE still performed best, but the differences in precision between the various estimators are much less clear-cut and the rate of improvement in $d$ is not as high as under known margins.

Our work specifically addressed the challenges of inference in large dimensions which is important for practical applications. Large dimensions up to $d=100$ were tackled for the first time and numerical challenges when working in such large dimensions were addressed in detail. Moreover, a detailed implementation of the presented estimation methods in the \R{} package \pkg{copula} creates transparency and allows the reader to access and verify our results.
\appendix

\section{Appendix}\label{sec.app}
\KOMAoptions{paper=landscape}%
\recalctypearea%
\areaset[current]{280mm}{130mm}
\begin{table}[htbp]
  \centering\scriptsize
\begin{tabularx}{\textwidth}{@{\extracolsep{\fill}}d{1}{2}cr%
d{3}{1}d{3}{1}%
d{3}{1}d{4}{1}%
d{5}{1}d{7}{1}%
d{4}{1}d{5}{1}%
d{4}{1}d{4}{1}%
d{4}{1}d{4}{1}%
d{5}{1}d{5}{1}%
d{2}{1}d{0}{1}%
d{3}{1}d{6}{1}%
d{4}{1}d{3}{1}}%
    \toprule
    \multicolumn{3}{c}{1000 Bias}&\multicolumn{20}{c}{Estimator}\\
	\cmidrule(lr{0.4em}){4-23}
	\multicolumn{1}{c}{$\tau$}&\multicolumn{1}{c}{Fam.}&$d$&\multicolumn{2}{c}{$\tau_{\bar{\hat{\tau}}}$}&\multicolumn{2}{c}{$\tau_{\bar{\hat{\theta}}}$}&\multicolumn{2}{c}{$\beta$}&\multicolumn{2}{c}{$\MDE{\chi}{CvM}$}&\multicolumn{2}{c}{$\MDE{\chi}{KS}$}&\multicolumn{2}{c}{$\MDE{\Gamma}{CvM}$}&\multicolumn{2}{c}{$\MDE{\Gamma}{KS}$}&\multicolumn{2}{c}{MLE}&\multicolumn{2}{c}{SMLE}&\multicolumn{2}{c}{DMLE}\\
    \midrule
 0.25 & A      & 5   & -4.3         &         (1.6)        &   -22.9        &           (8.7)        & -10.1        & (3.9)        &  -12.2        &          (4.6)        & -7.4         &         (2.8)        &  -54.3        &          (20.6)       & -824.7       &         (313.7)      & -2.6         & (1.0)        & -7.0         & (2.7)        & -35.8        & (13.6)       \\
       &        & 20  & -5.5         &         (-15.2)      &   -20.2        &           (-55.5)      & -42.9        & (-117.7)     &  -1.9         &          (-5.3)       & -3.8         &         (-10.4)      &  -38.3        &          (-105.0)     & -831.0       &         (-2277.7)    & 0.4          & (1.0)        & -0.4         & (-1.0)       & -27.9        & (-76.5)      \\
       &        & 100 & 0.9          &         (-9.5)       &   -21.9        &           (241.5)      & 171.3        & (-1888.9)    &  -0.0         &          (0.5)        & 0.1          &         (-1.0)       &  -39.2        &          (432.7)      & -827.3       &         (9122.9)     & -0.1         & (1.0)        & -0.6         & (6.7)        & -26.5        & (292.4)      \\
\bottomrule
  \end{tabularx}
  \caption{Biases (under known margins) multiplied by $1000$. The numbers in parentheses denote the factors of the corresponding entries with respect to the performance of the MLE.}
  \label{tab.bias}
\end{table}
\begin{table}[htbp]
  \centering\scriptsize
  \begin{tabularx}{\textwidth}{@{\extracolsep{\fill}}d{1}{2}cr%
d{4}{1}d{1}{1}%
d{4}{1}d{1}{1}%
d{5}{1}d{3}{1}%
d{4}{1}d{2}{1}%
d{4}{1}d{2}{1}%
d{4}{1}d{1}{1}%
d{5}{1}d{2}{1}%
d{3}{1}d{0}{1}%
d{4}{1}d{2}{1}%
d{4}{1}d{1}{1}}%
    \toprule
    \multicolumn{3}{c}{1000 RMSE}&\multicolumn{20}{c}{Estimator}\\
	\cmidrule(lr{0.4em}){4-23}
	\multicolumn{1}{c}{$\tau$}&\multicolumn{1}{c}{Family}&$d$&\multicolumn{2}{c}{$\tau_{\bar{\hat{\tau}}}$}&\multicolumn{2}{c}{$\tau_{\bar{\hat{\theta}}}$}&\multicolumn{2}{c}{$\beta$}&\multicolumn{2}{c}{$\MDE{\chi}{CvM}$}&\multicolumn{2}{c}{$\MDE{\chi}{KS}$}&\multicolumn{2}{c}{$\MDE{\Gamma}{CvM}$}&\multicolumn{2}{c}{$\MDE{\Gamma}{KS}$}&\multicolumn{2}{c}{MLE}&\multicolumn{2}{c}{SMLE}&\multicolumn{2}{c}{DMLE}\\
    \midrule
 0.25 & A      & 5   &     77.2     &             (1.7)    &       81.5     &               (1.8)    & 125.3    &     (2.7)    &      80.8     &              (1.8)    &     80.4     &             (1.8)    &      213.2    &              (4.6)    &     831.5    &             (18.1)   & 45.9     &    (1.0)    & 48.2     &     (1.0)    & 115.7    &     (2.5)    \\
       &        & 20  &     51.5     &             (3.1)    &       54.4     &               (3.3)    & 160.1    &     (9.7)    &      35.3     &              (2.1)    &     37.5     &             (2.3)    &      167.7    &              (10.2)   &     834.3    &             (50.7)   & 16.5     &    (1.0)    & 17.5     &     (1.1)    & 84.3     &     (5.1)    \\
       &        & 100 &     45.9     &             (10.0)   &       49.4     &               (10.8)   & 171.3    &     (37.4)   &      15.5     &              (3.4)    &     16.1     &             (3.5)    &      158.7    &              (34.6)   &     832.6    &             (181.6)  & 4.6      &    (1.0)    & 5.1      &     (1.1)    & 80.6     &     (17.6)   \\
\bottomrule
  \end{tabularx}
  \caption{RMSE (under known margins) multiplied by $1000$. The numbers in parentheses denote the factors of the corresponding entries with respect to the performance of the MLE.}
  \label{tab.rmse}
\end{table}
\begin{table}[htbp]
  \centering\scriptsize
  \begin{tabularx}{\textwidth}{@{\extracolsep{\fill}}d{1}{2}cr%
d{4}{1}d{1}{1}%
d{4}{1}d{1}{1}%
d{2}{1}d{0}{1}%
d{4}{1}d{2}{1}%
d{4}{1}d{2}{1}%
d{5}{1}d{2}{1}%
d{4}{1}d{2}{1}%
d{3}{1}d{0}{1}%
d{5}{1}d{2}{1}%
d{2}{1}d{0}{1}}%
    \toprule
    \multicolumn{3}{c}{User time in ms}&\multicolumn{20}{c}{Estimator}\\
	\cmidrule(lr{0.4em}){4-23}
	\multicolumn{1}{c}{$\tau$}&\multicolumn{1}{c}{Family}&$d$&\multicolumn{2}{c}{$\tau_{\bar{\hat{\tau}}}$}&\multicolumn{2}{c}{$\tau_{\bar{\hat{\theta}}}$}&\multicolumn{2}{c}{$\beta$}&\multicolumn{2}{c}{$\MDE{\chi}{CvM}$}&\multicolumn{2}{c}{$\MDE{\chi}{KS}$}&\multicolumn{2}{c}{$\MDE{\Gamma}{CvM}$}&\multicolumn{2}{c}{$\MDE{\Gamma}{KS}$}&\multicolumn{2}{c}{MLE}&\multicolumn{2}{c}{SMLE}&\multicolumn{2}{c}{DMLE}\\
    \midrule
 0.25 & A      & 5   &      6.4     &              (0.3)   &        10.2    &                (0.5)   & 3.1     &      (0.2)   &       64.7    &               (3.3)   &      90.1    &              (4.6)   &       70.1    &               (3.6)   &      96.5    &              (4.9)   & 19.5    &     (1.0)   & 3950.4  &      (202.1) & 3.4     &      (0.2)   \\
       &        & 20  &      58.6    &              (2.3)   &        142.5   &                (5.6)   & 3.8     &      (0.1)   &       277.5   &               (11.0)  &      389.7   &              (15.4)  &       299.1   &               (11.8)  &      416.0   &              (16.5)  & 25.2    &     (1.0)   & 3598.9  &      (142.6) & 3.8     &      (0.2)   \\
       &        & 100 &      1340.6  &              (19.1)  &        3498.9  &                (49.7)  & 6.8     &      (0.1)   &       2511.6  &               (35.7)  &      3534.4  &              (50.2)  &       3272.5  &               (46.5)  &      4743.4  &              (67.4)  & 70.4    &     (1.0)   & 4202.8  &      (59.7)  & 4.7     &      (0.1)   \\
\bottomrule
  \end{tabularx}
  \caption{Mean user run times (under known margins) in milliseconds. The numbers in parentheses denote the factors of the corresponding entries with respect to the performance of the MLE.}
  \label{tab.mut}
\end{table}
\begin{table}[htbp]
  \centering\scriptsize
  \begin{tabularx}{\textwidth}{@{\extracolsep{\fill}}d{1}{2}cr%
d{3}{1}d{3}{1}%
d{3}{1}d{4}{1}%
d{5}{1}d{7}{1}%
d{4}{1}d{5}{1}%
d{4}{1}d{4}{1}%
d{4}{1}d{4}{1}%
d{5}{1}d{5}{1}%
d{2}{1}d{0}{1}%
d{3}{1}d{6}{1}%
d{4}{1}d{3}{1}}%
    \toprule
    \multicolumn{3}{c}{1000 Bias}&\multicolumn{20}{c}{Estimator}\\
	\cmidrule(lr{0.4em}){4-23}
	\multicolumn{1}{c}{$\tau$}&\multicolumn{1}{c}{Fam.}&$d$&\multicolumn{2}{c}{$\tau_{\bar{\hat{\tau}}}$}&\multicolumn{2}{c}{$\tau_{\bar{\hat{\theta}}}$}&\multicolumn{2}{c}{$\beta$}&\multicolumn{2}{c}{$\MDE{\chi}{CvM}$}&\multicolumn{2}{c}{$\MDE{\chi}{KS}$}&\multicolumn{2}{c}{$\MDE{\Gamma}{CvM}$}&\multicolumn{2}{c}{$\MDE{\Gamma}{KS}$}&\multicolumn{2}{c}{MLE}&\multicolumn{2}{c}{SMLE}&\multicolumn{2}{c}{DMLE}\\
    \midrule
     0.25 & A      & 5   &    -4.3      &            (-1.6)    &      -22.9     &              (-8.4)    & -13.8     &    (-5.1)    &     34.9      &             (12.8)    &    41.8      &            (15.3)    &     10.4      &             (3.8)     &    -828.8    &            (-303.2)  & 2.7       &   (1.0)     & -1.3      &    (-0.5)    & -33.7     &    (-12.3)   \\
       &        & 20  &    -1.9      &            (-0.6)    &      -20.1     &              (-6.1)    & -64.5     &    (-19.6)   &     48.5      &             (14.7)    &    48.1      &            (14.6)    &     9.2       &             (2.8)     &    -826.7    &            (-251.1)  & 3.3       &   (1.0)     & 4.5       &    (1.4)     & -32.4     &    (-9.8)    \\
       &        & 100 &    -0.5      &            (0.1)     &      -22.3     &              (3.6)     & 125.2     &    (-20.0)   &     50.6      &             (-8.1)    &    53.1      &            (-8.5)    &     6.1       &             (-1.0)    &    -825.7    &            (131.6)   & -6.3      &   (1.0)     & -5.8      &    (0.9)     & -37.9     &    (6.0)     \\
\bottomrule
  \end{tabularx}
  \caption{Biases (under unknown margins) multiplied by $1000$. The numbers in parentheses denote the factors of the corresponding entries with respect to the performance of the MLE.}
  \label{tab.bias.pobs}
\end{table}
\begin{table}[htbp]
  \centering\scriptsize
  \begin{tabularx}{\textwidth}{@{\extracolsep{\fill}}d{1}{2}cr%
d{4}{1}d{1}{1}%
d{4}{1}d{1}{1}%
d{5}{1}d{3}{1}%
d{4}{1}d{2}{1}%
d{4}{1}d{2}{1}%
d{4}{1}d{1}{1}%
d{5}{1}d{2}{1}%
d{3}{1}d{0}{1}%
d{4}{1}d{2}{1}%
d{4}{1}d{1}{1}}%
    \toprule
    \multicolumn{3}{c}{1000 RMSE}&\multicolumn{20}{c}{Estimator}\\
	\cmidrule(lr{0.4em}){4-23}
	\multicolumn{1}{c}{$\tau$}&\multicolumn{1}{c}{Family}&$d$&\multicolumn{2}{c}{$\tau_{\bar{\hat{\tau}}}$}&\multicolumn{2}{c}{$\tau_{\bar{\hat{\theta}}}$}&\multicolumn{2}{c}{$\beta$}&\multicolumn{2}{c}{$\MDE{\chi}{CvM}$}&\multicolumn{2}{c}{$\MDE{\chi}{KS}$}&\multicolumn{2}{c}{$\MDE{\Gamma}{CvM}$}&\multicolumn{2}{c}{$\MDE{\Gamma}{KS}$}&\multicolumn{2}{c}{MLE}&\multicolumn{2}{c}{SMLE}&\multicolumn{2}{c}{DMLE}\\
    \midrule
     0.25 & A      & 5   &      77.2    &              (1.5)   &        81.5    &                (1.5)   & 127.9   &      (2.4)   &       70.7    &               (1.3)   &      74.6    &              (1.4)   &       101.4   &               (1.9)   &      833.8   &              (15.7)  & 53.0    &     (1.0)   & 53.0    &      (1.0)   & 93.7    &      (1.8)   \\
       &        & 20  &      52.3    &              (1.9)   &        54.3    &                (2.0)   & 185.6   &      (6.7)   &       60.8    &               (2.2)   &      59.7    &              (2.2)   &       89.5    &               (3.2)   &      832.2   &              (30.2)  & 27.6    &     (1.0)   & 29.3    &      (1.1)   & 83.5    &      (3.0)   \\
       &        & 100 &      44.5    &              (2.0)   &        49.2    &                (2.2)   & 125.2   &      (5.6)   &       58.5    &               (2.6)   &      60.1    &              (2.7)   &       93.1    &               (4.2)   &      831.3   &              (37.4)  & 22.2    &     (1.0)   & 24.7    &      (1.1)   & 75.8    &      (3.4)   \\
\bottomrule
  \end{tabularx}
  \caption{RMSE (under unknown margins) multiplied by $1000$. The numbers in parentheses denote the factors of the corresponding entries with respect to the performance of the MLE.}
  \label{tab.rmse.pobs}
\end{table}
\begin{table}[htbp]
  \centering\scriptsize
  \begin{tabularx}{\textwidth}{@{\extracolsep{\fill}}d{1}{2}cr%
d{4}{1}d{1}{1}%
d{4}{1}d{1}{1}%
d{2}{1}d{0}{1}%
d{4}{1}d{2}{1}%
d{4}{1}d{2}{1}%
d{5}{1}d{2}{1}%
d{4}{1}d{2}{1}%
d{3}{1}d{0}{1}%
d{5}{1}d{2}{1}%
d{2}{1}d{0}{1}}%
    \toprule
    \multicolumn{3}{c}{User time in ms}&\multicolumn{20}{c}{Estimator}\\
	\cmidrule(lr{0.4em}){4-23}
	\multicolumn{1}{c}{$\tau$}&\multicolumn{1}{c}{Family}&$d$&\multicolumn{2}{c}{$\tau_{\bar{\hat{\tau}}}$}&\multicolumn{2}{c}{$\tau_{\bar{\hat{\theta}}}$}&\multicolumn{2}{c}{$\beta$}&\multicolumn{2}{c}{$\MDE{\chi}{CvM}$}&\multicolumn{2}{c}{$\MDE{\chi}{KS}$}&\multicolumn{2}{c}{$\MDE{\Gamma}{CvM}$}&\multicolumn{2}{c}{$\MDE{\Gamma}{KS}$}&\multicolumn{2}{c}{MLE}&\multicolumn{2}{c}{SMLE}&\multicolumn{2}{c}{DMLE}\\
    \midrule
     0.25 & A      & 5   &      6.5     &              (0.3)   &        10.1    &                (0.5)   & 3.1     &      (0.1)   &       63.0    &               (3.0)   &      87.7    &              (4.1)   &       68.3    &               (3.2)   &      90.9    &              (4.3)   & 21.3    &     (1.0)   & 3966.3  &      (186.0) & 3.5     &      (0.2)   \\
       &        & 20  &      59.0    &              (2.4)   &        139.6   &                (5.6)   & 4.1     &      (0.2)   &       253.8   &               (10.2)  &      404.9   &              (16.3)  &       275.9   &               (11.1)  &      446.4   &              (18.0)  & 24.8    &     (1.0)   & 3685.1  &      (148.6) & 3.9     &      (0.2)   \\
       &        & 100 &      1339.5  &              (23.0)  &        3502.7  &                (60.2)  & 7.1     &      (0.1)   &       2666.4  &               (45.8)  &      3816.6  &              (65.6)  &       3336.2  &               (57.4)  &      4271.6  &              (73.4)  & 58.2    &     (1.0)   & 3964.5  &      (68.2)  & 4.4     &      (0.1)   \\
\bottomrule
  \end{tabularx}
  \caption{Mean user run times (under unknown margins) in milliseconds. The numbers in parentheses denote the factors of the corresponding entries with respect to the performance of the MLE.}
  \label{tab.mut}
\end{table}
\KOMAoptions{paper=portrait}%
\restoregeometry%
\recalctypearea%
\printbibliography[heading=bibintoc]
\end{document}